\documentclass[10pt,a4paper]{article}
\usepackage{jheppub}
\usepackage[utf8]{inputenc}
\usepackage[english]{babel}
\usepackage{amsmath}
\usepackage{amsfonts}
\usepackage{amssymb}
\usepackage{amsthm}
\usepackage{bm}

\def\l{\langle}
\def\r{\rangle}
\def\mb{\bar{m}}
\def\Mb{\bar{M}}
\def\qb{\bar{q}}
\def\pb{\bar{p}}
\def\qt{\tilde{q}}
\def\Qb{\bar{Q}}
\def\Pb{\bar{P}}
\def\pt{\tilde{p}}
\def\Pt{\tilde{P}}
\def\N{\mathcal{N}_1}
\def\Nn{\mathcal{N}_2}
\newcommand{\SOMMA}[2]{\sum\limits_{#1}^{#2}}
\newcommand{\resub}[1]{\textcolor{black}{#1}}

\newtheorem{remark}{Remark}
\newtheorem{theorem}{Theorem}
\newtheorem{proposition}{Proposition}

\newtheorem{definition}{Definition}
\newtheorem{lemma}{Lemma}
\newtheorem{assumption}{Assumption}
\usepackage[ruled,vlined]{algorithm2e}

\allowdisplaybreaks

\title{Serial vs parallel recall in the Blume-Every-Griffiths neural networks}
\author[a,b,1]{Linda Albanese,}
\author[b, c, 1]{Andrea Alessandrelli,}%
\author[d, 2, 3]{Adriano Barra,}%
\author[d,2]{Emilio N.M. Cirillo}%
\affiliation[a]{Istituto Nazionale di Alta Matematica Francesco Severi, Roma, Italy}
\affiliation[b]{Dipartimento di Matematica e Fisica, Università del Salento}
\affiliation[c]{Dipartimento di Informatica, Università di Pisa}
\affiliation[d]{Dipartimento di Scienze di Base e Applicazioni all'Ingegneria, Sapienza Università di Roma}
\affiliation[1]{Istituto Nazionale di Fisica Nucleare, Sezione di Lecce}
\affiliation[2]{Istituto Nazionale di Fisica Nucleare, Sezione di Roma}
\affiliation[3]{CNR Nanotec, Sezione di Lecce}
\abstract{Fully connected Blume-Emery-Griffiths neural networks performing pattern recognition and associative memory have been heuristically studied in the past (mainly via the replica trick and under the replica  symmetric assumption) as generalization of the standard Hopfield reference. In these notes, at first, by relying upon Guerra interpolation, we re-obtain the existing picture rigorously. Next we show that, due to dilution in the patterns, these networks are able to switch from serial recall (where one pattern is retrieved per time) to parallel recall (where several patterns are retrieved at once) and the larger the dilution, the stronger this emerging multi-tasking capability. In particular, we inspect the regimes of {\em mild dilution} (where solely a low storage of pattern can be enabled) and {\em extreme dilution} (where a medium storage of patterns can be sustained) separately as they give rise to different outcomes: the former displays hierarchical recall (distributing the amplitudes of the retrieved signals with different amplitudes), the latter --instead-- performs a equal-strength recall (where a $O(1)$ fraction of all the patterns is simultaneously retrieved with the same amplitude per pattern). 
\newline
Finally, in order to implement graded responses in the neurons,  variations on theme obtained by enlarging the possible values of neural activity these neurons may sustain are also discussed generalizing the Ghatak-Sherrington model for {\em inverse freezing} in Hebbian terms. }
\begin{document}
\maketitle

\section{Introduction}
Statistical mechanics of complex systems has long provided a fruitful framework for understanding pattern recognition, memory storage and information processing in neural networks \cite{Amit,Coolen} and, nowadays,  together with high-dimensional statistical inference, it is becoming a benchmark reference for theoretical inspections of modern Machine Learning algorithms and neural architectures \cite{BarraLenka,Carleo}. 
In particular, since the seminal contribution of Hopfield \cite{Hopfield}, Hebbian neural networks have been successfully recast in terms of spin-glass models, thereby enabling the use of powerful analytical tools originally developed in the study of disordered systems as well as importing the plethora of concepts  originally stemmed in that field (e.g. replicas, overlaps, aging,  etc.)  \cite{guerra_broken,MPV,talagrand2003spin}. 
\newline
Still within the research in disordered systems, in the past decades several models were proposed as spin glasses undergoing the phenomenon of {\em inverse freezing}\footnote{In a nutshell, {\em inverse freezing} may happen in systems whose  ground state entropy is higher than that associated to a positive temperature and this results in the (rather counterintuitive) phenomenon that a liquid object becomes solid upon heating, namely it undergoes inverse freezing.} \cite{NoiInverseFreezing,crisanti2005stable,crisanti-leuzziFullRSB_PRB} and these generalizations of the Sherrington-Kirkpatrick reference \cite{sherrington1975solvable}  motivated a bunch of extensions from the Hopfield ones \cite{Hopfield} also in the neural network's counterpart: indeed, starting from simple perceptrons up to fully connected neural networks, the statistical mechanics of these {\em inverse freezing}-driven neural networks has been studied extensively within the replica trick \cite{FreezingHopfield,InverseNew1,InverseNew2,InverseNew3}, and, in particular, under the so-called (and fairly standard in the field) replica-symmetric approximation \cite{BlumeCambelNN1,BlumeCambelNN2}: these analyses can be summarized in effective capacity-noise-dilution phase diagrams where different regions of this three-dimensional space highlight different computational capabilities of the underlying networks much as different regions of physical phase diagrams tackle different macroscopic behavior of networks of atoms or molecules (e.g. liquid, solid, gas) in the classical framework where statistical mechanics is used to describe the structure of matter.  
\newline
However, while the replica method has provided deep insights, its heuristic nature leaves open the question of whether these results can be placed on firm mathematical ground: in these regards, in the past two decades, rigorous methods have been developed to complement and in some cases replace replica-based reasoning (see e.g. \cite{bovier2012mathematical,talagrand2000rigorous}) and, within this bag of tools,  Guerra’s interpolation  has emerged as probably one of the most powerful and flexible technique\resub{s}, enabling exact results in a variety of disordered systems, including spin glasses \cite{barra2010replica,B-war1,B-war3,B-war4} and neural networks \cite{EliGuerra,GuerraNN,AABO-JPA2020,Albanese2021}. Applying such rigorous methods to these {\em inverse freezing}-driven neural networks is thus a natural step in consolidating the theoretical understanding of their associative and retrieval properties and it constitutes indeed the first part of the present investigation (where previous results on pattern recognition in the high storage regime have been punctually confirmed).

Beyond mathematical rigor, an equally important motivation for the existence of the present work is the exhaustive exploration of possible new computational capabilities eventually shown by these networks w.r.t. the classical pattern recognition (where one pattern is retrieved per time). In these regards, in the second part of the paper, we show how dilution of patterns deeply modifies the retrieval phenomenology (in suitable regions of the phase diagram)\footnote{This observation is rather general (i.e. it holds also for the Hopfield benchmark if forced to work with diluted patterns \cite{MultiTasking,Coolen1,Coolen2,FedeParallelo}), yet --at the time the above investigations on these networks were performed-- these multitasking capabilities nor were yet understood neither were present in the Literature.}.  
\newline
Indeed, as we prove in the second part of the manuscript, pattern's dilution actually allows the network to switch between serial recall  and parallel recall, the latter becoming stronger as dilution increases\footnote{Stronger in the sense that, the higher the dilution, the larger the amount of different patterns that can be simultaneously held by the machine with a transition from a hierarchical ordering of signals (for mild dilution) toward a regime (for extreme dilution) where the signals pertaining to all the patterns are simultaneously raised with the same strength.}, highlighting a close connection between sparsity of representation and distributed processing in such networks.

A further variation on theme that we face here lies in studying also models equipped with graded response neurons: in the Hopfield model, neurons are reduced to their most basic representation, namely they are stylized as binary variables (i.e., noisy on-off switches accounted by Ising spins or Boolean variables). While this strong approximation is enough to capture the essential mechanism for neural dialogues,  the quest for networks equipped with graded response neurons soon raised in the Community \cite{Treveso}  (even by Hopfield himself \cite{HopfieldGraded}\footnote{We remark that, another over-simplification  assumed in this minimal stylization of neurons is the total absence of their dendritic arbor, whose contribution  toward a graded response by the neuron can be relevant \cite{Kanter1,Kanter2}.}) and, to overcome the binary approximation, a natural line of generalization  paved in the past consisted in the shift from Ising spins to Potts spins. While certainly implementing  multi-state neurons via Potts variables in the Hopfield model \textit{à la Kanter} \cite{KanterPotts} is a way out, another alternative route lies in the graded neural response offered by the spin variables used in these models for inverse freezing \cite{ghatak1977crystal,BlumeCambelNN1} as  they have spins taking several values, ranging from $\{-1,0,+1\}$ of the Blume-Every-Griffith model to more fragmented partitions of the domain $[-1,+1]$  of the Ghatak-Sherrington model (for a review where all these models are analyzed at once see \cite{NoiInverseFreezing}).  
 
Therefore, purpose of this work is to investigate two particular neural network models inspired by the \emph{inverse freezing} Literature.  
The first is the Blume-Emery-Griffiths (BEG) model \cite{BlumeCambelNN2}, where both neurons and pattern entries can take values in $\{-1,0,+1\}$.  
The second is a generalization of the BEG model equipped with graded response neurons, namely the Ghatak-Sherrington (GS) model (originally proposed to study inverse freezing in \cite{katayama1999ghatak} but here equipped with Hebbian couplings). Our aim is to understand, for both models, how their information processing capabilities are affected by the interplay between \emph{pattern dilution} and \emph{pattern storage}.  

The paper is structured as follows: in Sec ~\ref{sec:generalities} we introduce the two models and their related  control and order parameters required to allow for a statistical mechanical treatment. Then, in Sec.~\ref{sec:serial} we inspect standard pattern recognition (i.e. serial retrieval) in these models forced to work in the high storage regime: purpose of this Section is to confirm existing know-how heuristically obtained in the past. 
\newline
After that, in Sec.~\ref{sec:lowloadS} we turn our focus toward the, so far not-explored, regime where an optimal tradeoff among strong enough dilution and suitably reduced pattern load  open the possibility of \emph{parallel retrieval}:  in this scenario, under suitable conditions that stabilize these multi-tasking capabilities, the network can sustain the retrieval of several patterns simultaneously, raising multiple Mattis magnetizations at once\footnote{For the sake of clearness, we highlight that this multitasking scenario \cite{MultiTasking,Coolen1,Coolen2} has nothing to share with a Hopfield model trapped into a spurious states \cite{Amit,Bovier1,Bovier2}: the reasons that allow to raise multiple magnetizations in these two cases are totally different, hence the two mechanisms -parallel retrieval and metastability- have not to be confused.}.

\section{Generalities}
\label{sec:generalities}

In this section we introduce the two models studied in this work, namely the Blume-Emery-Griffiths (BEG) and the Ghatak-Sherrington (GS) extensions of the classical Hopfield model for pattern recognition. We then discuss the two computational regimes relevant for these architectures --serial and parallel-- highlighting how they affect the dynamics and the interpretation of retrieval. Finally, we collect the main tools from the statistical mechanics of disordered systems that will be employed in the rest of the manuscript.

\subsection{The Blume-Emery-Griffiths neural network}
\label{subsec:BEG_model}

We start by introducing the Hamiltonian of the BEG model together with the standard statistical mechanics framework.

\begin{definition}[BEG Hamiltonian]
\label{def:probBEG}
Consider a neural configuration $\boldsymbol{\sigma}\in \boldsymbol{\Omega}^N \equiv \{0,\pm1\}^N$ and a set of $K$ ($N$-bits long) patterns $\boldsymbol{\xi}=\{\xi_i^\mu\}$ whose entries distribution is
\begin{align}
\mathbb{P}(\xi_i^\mu)
=\frac{a}{2}\,\delta_{\xi_i^\mu,+1}
+\frac{a}{2}\,\delta_{\xi_i^\mu,-1}
+(1-a)\,\delta_{\xi_i^\mu,0},
\label{eq:dist_xi}
\end{align}
with $a\in[0,1]$\footnote{This means that,  for general $a \in (0,1)$, patterns bring information (i.e. a value $\pm 1$) but they can also have uninformative entries (i.e. a value $0$)  while for $a=1$ patterns are no longer diluted and in the $a\to 0$ limit they yields completely blank entries (hence there is no information to retrieve and no network required for such a task).}. 
\newline
The Hamiltonian (or \emph{cost function} in machine-learning language) of the BEG neural network model is \cite{blume1971ising}
\begin{align}
\mathcal H^{(BEG)}_{N,K,a}(\boldsymbol{\sigma}\mid \boldsymbol{\xi})
&=-\frac{1}{2Na}\sum_{i,j=1}^N\sum_{\mu=1}^K \xi_i^\mu\xi_j^\mu\,\sigma_i\sigma_j
-\frac{1}{2Na(1-a)}\sum_{i,j=1}^N\sum_{\mu=1}^K \eta_i^\mu\eta_j^\mu\,\sigma_i^2\sigma_j^2,
\label{eq:Hamiltonian}
\end{align}
where for $i=1,\dots,N$ and $\mu=1,\dots,K$ we define
\begin{equation}
    \eta_i^\mu:=(\xi_i^\mu)^2-a.
    \label{eq:def_eta}
\end{equation} 
\end{definition}

\begin{remark}[Normalization and basic statistics]
The prefactors\textcolor{black}{, and the variables in general,} in \eqref{eq:Hamiltonian} are chosen so that the synaptic couplings have \textcolor{black}{null mean and} unitary variance. Indeed,
\begin{align}
\mathbb{E}[\xi_i^\mu]=0,\qquad \mathrm{Var}(\xi_i^\mu)=a,
\end{align}
and
\begin{align}
\mathbb{E}[\eta_i^\mu]=0,\qquad \mathrm{Var}(\eta_i^\mu)=a(1-a).
\end{align}
Although $(\boldsymbol{\xi},\boldsymbol{\eta})$ are deterministically related and hence dependent, their covariance vanishes:
\begin{align}
\mathrm{Cov}(\xi_i^\mu,\eta_i^\mu)
=\mathbb{E}\!\left[\xi_i^\mu\big((\xi_i^\mu)^2-a\big)\right]
=\mathbb{E}[(\xi_i^\mu)^3]-a\,\mathbb{E}[\xi_i^\mu]=0,
\label{eq:cov}
\end{align}
by symmetry of the distribution of $\xi_i^\mu$.
\end{remark}

\resub{To quantitatively describe the behavior of the model, we recall the following standard definitions.}
\resub{\begin{definition}[Partition function]
\label{def:partitionfunction}
For a given inverse temperature $\beta=1/T$ and considered the BEG Hamiltonian $\mathcal H$, the partition function of the model \eqref{eq:Hamiltonian} is
\begin{equation} \label{ZCW}
 \mathcal{Z}^{(BEG)}_{N,K,a}(\beta\mid \boldsymbol{\xi})
 := \sum_{\{\boldsymbol \sigma\}} 
 e^{-\beta \mathcal H^{(BEG)}_{N,K,a}(\boldsymbol \sigma\mid\boldsymbol{\xi})}
 = \sum_{\{\boldsymbol \sigma\}} \mathcal B^{(BEG)}_{N,K,a}(\beta,\boldsymbol\sigma\mid\boldsymbol{\xi}),
\end{equation}
where $\mathcal B^{(BEG)}_{N,K,a}(\beta,\boldsymbol\sigma\mid\boldsymbol{\xi})$ denotes the Boltzmann weight.  
Throughout the paper, the notation $\sum\limits_{\{\boldsymbol\sigma\}}$ indicates the sum over all neural configurations $\boldsymbol\sigma\in\Omega=\{-1,0,+1\}^N$.
\end{definition}}
\resub{\begin{definition}[Gibbs measure]
The Gibbs expectation of a generic (smooth) observable $f(\boldsymbol\sigma)$ is
\begin{equation}   \label{BMCW}
 \omega_{N,K,a}^{(BEG),\boldsymbol{\xi}}(f(\boldsymbol\sigma))
 := \frac{\sum_{\{\boldsymbol\sigma\}} f(\boldsymbol\sigma)\, 
 \mathcal B^{(BEG)}_{N,K,a}(\beta,\boldsymbol\sigma\mid\boldsymbol{\xi})}
 {\mathcal{Z}^{(BEG)}_{N,K,a}(\beta\mid\boldsymbol{\xi})}.
\end{equation}
\end{definition}}
\resub{\begin{definition}[Replicated system]
    A \emph{replica} is a copy of the system with exactly the same realization of the disorder $\bm J = \bm \xi$ \cite{MPV}, i.e all the weights must coincide among two replicas of the same network.\\
    More generally, we can introduce an arbitrary number $n$ of independent replicas of the system, all subject to the same realization of the disorder  $\bm J =\bm \xi$ and characterized by neural configurations $\boldsymbol \sigma^{(1)}, \ldots, \boldsymbol \sigma^{(n)}$. Their joint distribution over the replicated neurons is given by the product state
\begin{equation}
    \label{eq:Gibbstotalaver}
    \Omega_{N,K,a}^{(BEG),\boldsymbol J} = (\omega_{N,K,a}^{(BEG),\boldsymbol J})^{(1)} \times (\omega_{N,K,a}^{(BEG),\boldsymbol J})^{(2)} \times \cdots \times (\omega_{N,K,a}^{(BEG),\boldsymbol J})^{(n)},
\end{equation}
where each factor $(\omega_{N,K,a}^{(BEG),\boldsymbol J})^{(l)}$ acts on the  neurons $\sigma_i^{(l)}$ of replica $l$.  
\end{definition}}

\resub{\begin{definition}[Global average]
\label{def:global}
The global average of a generic observable $f$, depending on the configurations of $n$ replicas as well as on their weights, is
\begin{equation}
\label{eq:totalaver}
    \langle f(\boldsymbol \sigma^{(1)},\boldsymbol  \sigma^{(2)}, \ldots,\boldsymbol  \sigma^{(n)}) \rangle_{N,K,a}
    = \mathbb{E}_{\boldsymbol\xi}\,\Big[ \Omega_{N,K,a}^{(BEG),\boldsymbol \xi}\big(f(\boldsymbol \sigma^{(1)}, \boldsymbol \sigma^{(2)}, \ldots, \boldsymbol \sigma^{(n)})\big)\Big],
\end{equation}
where the expectations $\mathbb{E}_{\boldsymbol\xi}$ average over the quenched patterns.
\end{definition}}

\resub{\begin{definition}[Quenched statistical pressure / Quenched free energy]
\label{def:quenchedfreeenergy}
The quenched (intensive) statistical pressure related to the BEG neural network is defined as
\begin{equation}
\mathcal A^{(BEG)}_a(\beta, \alpha)
 := \lim_{N\to\infty} \frac{1}{N}\,
 \mathbb{E}_{\boldsymbol\xi}\,
 \ln \mathcal{Z}^{(BEG)}_{N,K,a}(\beta\mid\boldsymbol{\xi}),
\end{equation}
and relates to the quenched (intensive) free energy via $\mathcal A^{(BEG)}_a(\beta, \alpha)=-\beta \mathcal F^{(BEG)}_a(\beta, \alpha)$. 
\end{definition}}
\resub{Note that, for historical reasons, we will study the statistical pressure (more standard at work with interpolation techniques) but clearly we can obtain the same identical information by studying the free energy (more standard at work with the replica methods).}

We are interested in the collective behavior of the system rather than in single neuron statistics, hence we introduce suitable \emph{order parameters}, namely macroscopic observables able to capture  the organization of the network as the control parameters are tuned: in these regards, it is useful to introduce the two next

\begin{definition}[Control parameters]\label{Def:ControlParameters}
The natural control parameters of these neural networks are 
\begin{itemize}
    \item  the fast (thermal) noise, tuned by $\beta \equiv T^{-1} \in \mathbb{R}^+$.
    \item the pattern entry's dilution, ruled by  $a \in [0,1]$.
    \item the pattern storage capacity (or network load) $\alpha$,  defined as
\begin{align}
\label{eq:load}
\alpha := \lim_{N\to\infty} \frac{K}{N^{\delta}}, \ \ \ \delta \in (0,1].
\end{align}
\end{itemize}
\end{definition}
In particular, when the amount of patterns scales at most as $K \propto \ln N$ we informally speak about the {\em low storage} regime (here $\alpha=0$ and $\delta$ does not play any role), for $K \sim N^{\delta}$ and $0<\delta<1$ we informally speak about the {\em medium storage} regime  (here $\alpha=0$) and, finally, for $K \propto N^{\delta}$ and $\delta=1$ we speak of the {\em high storage} regime   (here $\alpha>0$). 
  
\begin{definition}[Order parameters]\label{Def:OrderParameters}
The natural order parameters of these neural networks are the Mattis magnetizations with respect to $\boldsymbol{\xi}$ and $\boldsymbol{\eta}$:
\begin{align}
m_\mu(\boldsymbol{\sigma})
:=\frac{1}{Na}\sum_{i=1}^N \xi_i^\mu\,\sigma_i,  \ \ \ \mu \in (1,...,K),
\label{eq:m_mattis_order}
\end{align}
which measure the quality of retrieval of pattern $\mu$, and
\begin{align}
M_\mu(\boldsymbol{\sigma})
:=\frac{1}{Na(1-a)}\sum_{i=1}^N \eta_i^\mu\,\sigma_i^2,  \ \ \ \mu \in (1,...,K),
\label{eq:M_order}
\end{align}
which measure the alignment of the activity field $\boldsymbol{\sigma}^2$ with $\boldsymbol{\eta}^\mu$.\footnote{Solely in the high-storage investigation $\alpha > 0$, deepened only in the first part of the manuscript (where our aim is simply to recover the heuristic picture already achieved via the replica trick), we will need also the spin glass order parameters, but -for the sake of simplicity- we will introduce them just when we will need them.}
\end{definition}

Rewriting the Hamiltonian in terms of these parameters yields
\begin{align}
\mathcal H_{N,K,a}(\boldsymbol{\sigma}\mid \boldsymbol{\xi})
=-\frac{N}{2}\,a\sum_{\mu=1}^K m^2_\mu(\boldsymbol{\sigma})
-\frac{N}{2}\,a(1-a)\sum_{\mu=1}^K M_\mu^2(\boldsymbol{\sigma}),
\label{eq:Hamiltonian_mM}
\end{align}
which will be the starting point of the analysis in Subsec.~\ref{subsec:stat_mecc_app_BEG}.

\begin{remark}[On the non-dense character of the model]
Although the Hamiltonian \eqref{eq:Hamiltonian} contains quartic terms, the model cannot be regarded as a dense network. The $P=4$ quartic interactions appear only in the form $\sigma_i^2\sigma_j^2$ (self-interaction of spin activities), rather than genuine multi-spin terms such as $\sigma_{i_1}\sigma_{i_2}\cdots\sigma_{i_p}$.  
As a consequence, the storage capacity of this network scales linearly with $N$ (i.e. $K\propto N$), unlike truly dense $P$-spin models where the capacity grows in a polynomial way as $K\propto N^{P-1}$ \cite{EliGuerra,Albanese2021,HopKro1}.
\end{remark}

In particular our goal is to express the statistical pressure in terms of the control and order parameters of the theory so to extremize it w.r.t the order parameters: the extremization procedure results in a series of self-consistent equations for the order parameters that rule their evolution in the space of the control parameters and whose inspection finally allows to obtain the phase diagram of the network ({\em vide infra}).

\subsection{The Ghatak--Sherrington neural network}
\label{subsec:GS_model}

Motivated by the study of spin-glass models for \emph{inverse freezing}, where such extensions were first introduced \cite{crisanti2005stable,crisanti-leuzziFullRSB_PRB,Kaufman1990,NoiInverseFreezing}, we now focus on a generalization of the previous setting: in this variant, known as the Ghatak-Sherrington (GS)  model \cite{ghatak1977crystal}, spins are allowed to take values in a larger, Potts-like, state space, namely
\begin{align}
\sigma \in \Omega=\left\{-1+\frac{k}{S}\; :\; k=0,\ldots,2S\right\}, 
\qquad 2S+1\in\mathbb{N}, \quad S\geq 1/2.
\end{align}
Note that for $S=1$ neurons in the GS network reduce to those of the BEG one, i.e.\ $\bm\sigma\in\{-1,0,+1\}^N$ and for $S=1/2$ they collapse  to standard Ising spins, i.e.\ $\bm\sigma\in\{-1,+1\}^N$.
This extension was originally considered by Ghatak and Sherrington \cite{ghatak1977crystal} and deepened later by Katayama and Horiguchi \cite{katayama1999ghatak} in the context of {\em inverse freezing}:  in the present context instead, this variability in the neural states plays the role of endowing neurons with a \emph{graded response}, rather than a purely binary one: the interest for graded response neurons (largely investigated along these decades, see e.g. \cite{Treveso,KanterPotts,Raimer91,Kanter1,kanter,Kanter2}) thus drives the analysis of a GS model whose interactions are now made Hebbian.
\newline
On remark, we stress that --as the neural activities gets scattered over several values-- the same fragmentation is implemented in the pattern entries too, as stated by the next
\begin{definition}
The distribution of pattern entries of the GS neural network model $\bm \xi= \{ \xi^\mu_i\}$ is defined as
\begin{align}
\mathbb{P}(\xi_i^\mu) =
\begin{cases}
(1-a)\,\delta_{\xi_i^\mu,0}
+\dfrac{a}{2S}\sum\limits_{\substack{k=0\\k\neq S}}^{2S}
\delta_{\xi_i^\mu,\,-1+\tfrac{k}{S}} & \text{if $S\in\mathbb{N}$}, 
\\\\
(1-a)\,\delta_{\xi_i^\mu,0}
+\dfrac{a}{2S+1}\sum\limits_{k=0}^{2S}
\delta_{\xi_i^\mu,\,-1+\tfrac{k}{S}} & \text{if $\tfrac{2S+1}{2}\in\mathbb{N}$},
\end{cases}
\label{eq:probdistrS}
\end{align}
with $a\in[0,1]$ as before (see Definition \ref{def:probBEG}).  

\end{definition}

We highlight that for $S=1$ this reduces to the BEG distribution (see Eq. \eqref{eq:dist_xi} in Def. \ref{def:probBEG}). \textcolor{black}{Moreover, we stress that, since our purpose is to explore the parallel processing, despite the absence of the $0$ in the state space when $(2S+1)/2 \in \mathbb{N}$, the patterns can also assume this value.}

\begin{definition}[GS Hamiltonian]
Let $\boldsymbol{\sigma}\in\Omega^N$ be a neural configuration and $\bm \xi=\{\xi_i^\mu\}$ the pattern entries distributed as in \eqref{eq:probdistrS}.  
The Hamiltonian of the GS neural network model is
\begin{align}
\mathcal H^{(GS)}_{N,K,a,S}(\boldsymbol\sigma \mid \boldsymbol\xi)
&=-\frac{1}{2N\,\N(a,S)}\sum_{i,j=1}^N\sum_{\mu=1}^K \xi_i^\mu\xi_j^\mu\,\sigma_i\sigma_j
-\frac{1}{2N\,\Nn(a,S)}\sum_{i,j=1}^N\sum_{\mu=1}^K \eta_i^\mu\eta_j^\mu\,\sigma_i^2\sigma_j^2,
\label{eq:HamiltSspin}
\end{align}
where we introduced
\begin{align}
\N(a,S) := \mathrm{Var}[\xi_i^\mu], 
\qquad 
\Nn(a,S) := \mathrm{Var}[\eta_i^\mu],
\label{eq:GS_normalization}
\end{align}
and, as in the BEG model,
\begin{align}
\eta_i^\mu := (\xi_i^\mu)^2 - \N(a,S).
\end{align}
\end{definition}
The explicit expressions of the normalizations $\N(a,S)$ and $\Nn(a,S)$ for a generic value of $a$ and $S$ are reported in Appendix \ref{app:momenta}.
\begin{remark}[Basic statistics]
By symmetry of the pattern distribution we still have
\begin{align}
\mathbb{E}[\xi_i^\mu]=\mathbb{E}[\eta_i^\mu]=0,
\qquad
\mathrm{Cov}(\xi_i^\mu,\eta_i^\mu)=\mathbb{E}[\xi_i^\mu\eta_i^\mu]
=\mathbb{E}[(\xi_i^\mu)^3]-\N(a,S)\,\mathbb{E}[\xi_i^\mu]=0.
\end{align}
\end{remark}

While the control parameters of the GS neural network are still those provided in Definition \ref{Def:ControlParameters} (that is, the noise $\beta$, the dilution $a$ and the patterns storage $\alpha$ fine tuned by  $\delta$), the order parameters slightly differ from those introduced in Definition \ref{Def:OrderParameters} --despite their information content is preserved-- as they read accordingly to the next
\begin{remark}[Order parameters]
For the GS neural network model is more convenient to define the Mattis magnetizations as
\begin{align}
m_\mu(\boldsymbol\sigma)
&:=\frac{1}{N\,\N(a,S)}\sum_{i=1}^N \xi_i^\mu\sigma_i,
&
M_\mu(\boldsymbol\sigma)
&:=\frac{1}{N\,\Nn(a,S)}\sum_{i=1}^N \eta_i^\mu\sigma_i^2,
\end{align}
\label{rem:order}
\end{remark}
As a result, the GS Hamiltonian \eqref{eq:HamiltSspin} can be rewritten in terms of these order parameters as
\begin{align}\label{FormaQuadratica}
\mathcal H^{(GS)}_{N,K,a,S}(\boldsymbol\sigma \mid \boldsymbol\xi)
=-\frac{N}{2}\,\N(a,S)\sum_{\mu=1}^K m_\mu^2(\boldsymbol\sigma)
-\frac{N}{2}\,\Nn(a,S)\sum_{\mu=1}^K M_\mu^2(\boldsymbol\sigma).
\end{align}

From now on, for notational convenience, we omit the dependence on $a$ and $S$ in $\N(a,S)$ and $\Nn(a,S)$ and on $\bm \sigma$ from $m_\mu$ and $M_\mu$. Moreover, we shall omit the superscripts $(BEG)$ and $(GS)$, as the context will always make clear which model is being considered, keeping in mind that the BEG model is recovered as the special case of the GS model with $S=1$.

\textcolor{black}{We stress that the statistical mechanics tools introduced in Defs. \ref{def:partitionfunction}-\ref{def:quenchedfreeenergy} are the same in this model. }

\subsection{Serial versus parallel retrieval}
\label{subsec:serial_vs_parallel_overview}

\begin{figure}[t!]
    \centering
    \includegraphics[width=15cm]{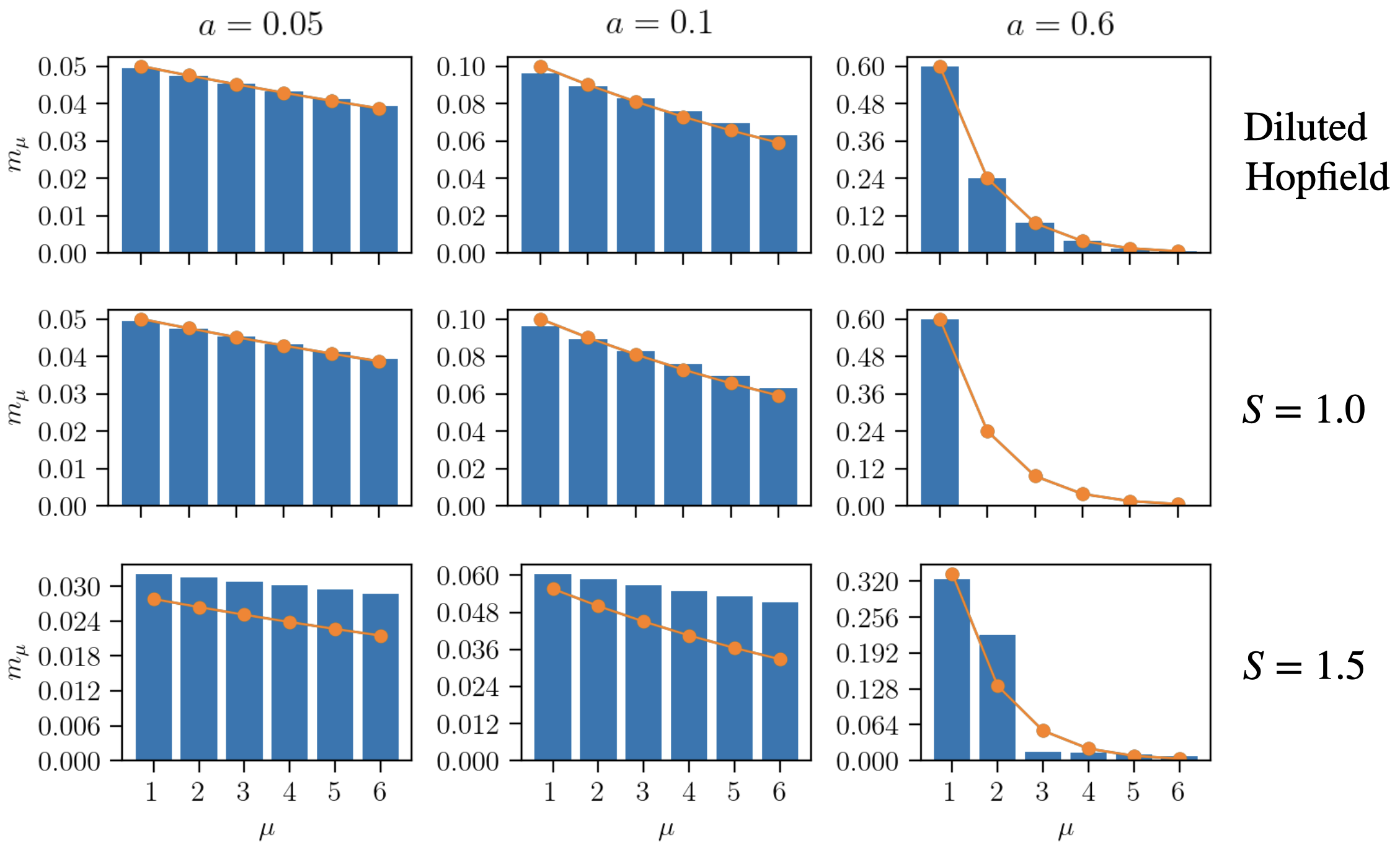}
    \caption{\resub{MCMC outcomes: Mattis magnetizations for \(N=1000\) and \(K=6\) diluted patterns, shown at \(\beta=1000\) for three activity levels: \(a=0.05\) (left), \(a=0.1\) (center), and \(a=0.6\) (right). Rows correspond to: \emph{(top)} diluted Hopfield, \emph{(middle)} BEG, and \emph{(bottom)} GS with \(S=1.5\).
    Bars report the steady-state overlaps \(m_\mu\) measured in Monte Carlo simulations, while the orange dotted line shows the hierarchical prediction \(m_\ell=\N(1-a)^{\ell-1}\) in \eqref{eq:hierarc_ansatz}.
    The diluted Hopfield model closely follows the hierarchical profile for all \(a\). When the additional \(\eta\)-sector is present (BEG and GS), deviations from the hierarchical ordering grow as dilution is reduced (larger \(a\)), and the number of non-vanishing overlaps drops more abruptly than in the diluted Hopfield case. This is explained by the energetic competition discussed in Sec.~\ref{subsec:eta_term}.}}
    \label{fig:gerarc_Hopf}
\end{figure}

\resub{
Before delving into the mathematical analysis, it is useful to outline the qualitative information-processing regimes that these networks can display as the control parameters are varied.
The relevant knobs are thermal noise (\(\beta\)), dilution (\(a\)), storage scaling (\(K\sim N^\delta\)), and the graded-response parameter \(S\) (i.e.\ the structure of the neural state space \(\Omega\)).
A convenient diagnostic is the set of Mattis magnetizations \(\{m_\mu\}_{\mu=1}^K\):
\begin{itemize}
    \item \emph{serial recall}: a single \(O(1)\) magnetization is present (one pattern is retrieved at a time);
    \item \emph{parallel recall}: several magnetizations are simultaneously non-vanishing (multitasking retrieval).
\end{itemize}
}

\noindent

\resub{
In the present context it is convenient to organize the phenomenology into three broad regimes:}
\resub{
\begin{itemize}
    \item \textbf{Ergodic (no-recall) regime.}
    At high temperature (small \(\beta\)) thermal fluctuations dominate and the overlaps remain close to zero, \(m_\mu\simeq 0\) for all \(\mu\).
    \item \textbf{Serial retrieval (pure state).}
    In the high-storage setting (\(\delta=1\), \(\alpha>0\)) and for finite \(a\in(0,1)\), the interference generated by an extensive number of stored patterns acts as a quenched noise. This typically stabilizes a \emph{pure} retrieval configuration, where one condensed pattern carries the signal while the remaining overlaps stay negligible.
    \item \textbf{Parallel retrieval (hierarchical or fully parallel states).}
    Away from high storage (\(\delta<1\), hence \(\alpha=0\)), dilution becomes a computational resource: since each pattern is informative only on an \(a\)-fraction of sites, retrieving one pattern leaves a \((1-a)\)-fraction of effectively unconstrained neurons that can align with other patterns.
    This opens the door to multitasking states \cite{FedeParallelo,Coolen1,Coolen2,MultiTasking}, whose structure depends on the dilution level:
    \begin{itemize}
        \item at \emph{mild dilution} (larger \(a\)), one typically observes a \emph{hierarchical} multitasking state with ordered overlaps,
        \begin{equation}
            m_\ell = \N(1-a)^{\ell-1},
            \qquad \ell=1,\dots,\hat K,
            \label{eq:hierarc_ansatz}
        \end{equation}
        where \(\hat K\) is an effective cutoff set by resource exhaustion (and finite-\(N\) discreteness), yielding \(\hat K\lesssim \log N\) at fixed \(a\).
        Figure~\ref{fig:gerarc_Hopf} provides a direct visualization of this progressive overlap organization.
        \item at \emph{strong dilution} (small \(a\)), the hierarchy progressively flattens until several non-zero overlaps become comparable, leading to a genuinely \emph{fully parallel} mode with
        \begin{equation}
            \bm m = \bar m\,(1,1,1,\ldots)\,.
        \end{equation}
    \end{itemize}
\end{itemize}}

\begin{figure}[t!]
    \centering
    \includegraphics[width=15cm]{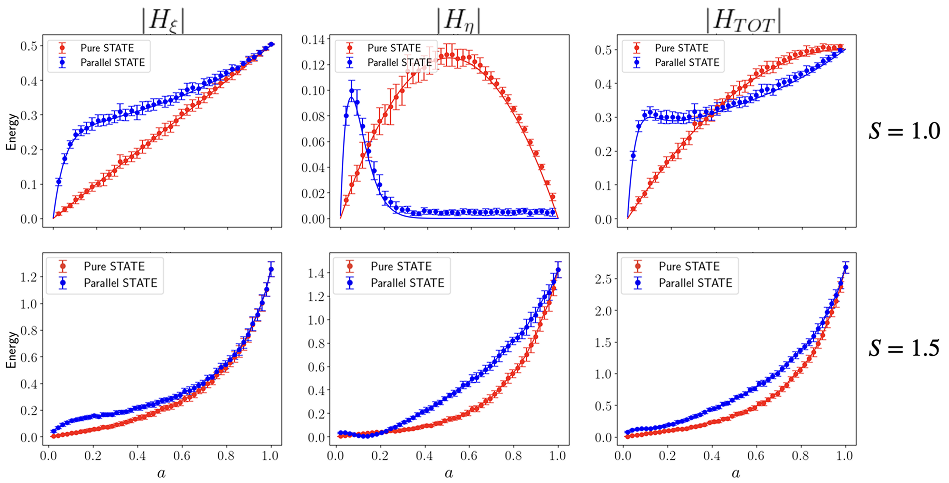}
    \caption{\resub{Energy comparison between the pure state \eqref{eq:pure_state} and the hierarchical/parallel state \eqref{eq:parallel_state}. Solid lines are the analytical predictions \eqref{eq:num_puro}--\eqref{eq:num_ger}, while dots are numerical simulations with \(N=2000\) and \(K=10\).
    Top row: \(S=1.0\), \(\bm{\sigma}\in\{-1,0,1\}^{N}\). Bottom row: \(S=1.5\), \(\bm{\sigma}\in\{-1,-1/3,1/3,1\}^{N}\).
    Columns show: (i) the \(\xi\)-dependent contribution \(\big|\mathcal H^{(\xi)}\big|\), (ii) the \(\eta\)-dependent contribution \(\big|\mathcal H^{(\eta)}\big|\), and (iii) the total energy.
    The \(\xi\)-term always favors the hierarchical/parallel state, while the \(\eta\)-term can favor the pure state at weak dilution when \(0\in\Omega\) (integer \(S\)), producing a crossover in the total energy for \(S=1\); by contrast, for \(S=1.5\) (half-integer \(S\)) the parallel state remains energetically preferred across essentially the whole range of \(a\).}}
    \label{fig:energy_comparison}
\end{figure}

\noindent
\resub{
A simple energetic argument explains why dilution can trigger a shift from serial to parallel processing.
Focusing on the retrieval-driving (quadratic) part of the cost function, both the BEG and GS models contain a term that, at fixed \(N\), can be written schematically as
\[
\tilde H(\{m_\mu\}) \propto -\frac{N}{2}\sum_{\mu=1}^{K} m_\mu^2,
\]
see \eqref{eq:Hamiltonian_mM} and the first term in \eqref{FormaQuadratica}.
Specialize to \(K=2\), so \(\tilde H\propto -\frac{N}{2}(m_1^2+m_2^2)\), and consider two complementary diluted patterns, informative on disjoint halves of the system,
\begin{equation}
\boldsymbol \xi^1=\big(\underbrace{\xi_1^1, \ldots, \xi_{N/2}^1}_{\in \{-1,-1+1/S,\ldots,+1\}^{N/2}},\underbrace{0,\ldots,0}_{\in \{0\}^{N/2}}\big),\qquad
\boldsymbol \xi^2=\big(\underbrace{0,\ldots,0}_{\in \{0\}^{N/2}},\underbrace{\xi_{N/2+1}^2,\ldots,\xi_N^2}_{\in \{-1,-1+1/S,\ldots,+1\}^{N/2}}\big).
\end{equation}
If the network retrieves only \(\boldsymbol\xi^1\), it aligns the first \(N/2\) neurons while the rest behave essentially at random, giving \(\mathbb{E}[m_1]\simeq \N/2\) and \(\mathbb{E}[m_2]\simeq 0\), hence
\[
\mathbb{E}[\tilde H]\propto -\frac{N\N^2}{8}.
\]
If it also exploits the remaining degrees of freedom to align with \(\boldsymbol\xi^2\), then \(\mathbb{E}[m_1]\simeq \mathbb{E}[m_2]\simeq \N/2\), so that
\[
\mathbb{E}[\tilde H]\propto -\frac{N\N^2}{4},
\]
which is lower. Thus, when patterns contain blank entries, spreading the alignment over multiple partially informative patterns can be energetically favorable, naturally promoting parallel retrieval.
\noindent
At the level of microscopic configurations, the hierarchical multitasking state can be expressed as \cite{MultiTasking}
\begin{equation}
\sigma_i^{(h)}=\xi_i^1+\sum_{\nu=2}^{\hat K}\xi_i^\nu \prod_{\rho=1}^{\nu-1}\delta_{\xi_i^\rho,0}\,,
\label{eq:parallel_state}
\end{equation}
i.e.\ each site aligns with the first pattern that is non-blank at that site.
This organization is stable down to a critical activity \(a_c\), where the leading overlap becomes comparable with the cumulative sub-leading signal, \(m_1\sim \sum_{k>1} m_k\) \cite{MultiTasking}, and the system crosses over toward the fully parallel configuration described above.
\noindent
Finally, in medium-load scenarios (\(K\sim N^\delta\) with \(\delta\in(0,1)\)) an extensive number of simultaneous recalls can persist only if dilution scales with \(N\), i.e.\ \(a=a(N)\sim N^{-\gamma}\), with \(\gamma\) tuned so that the effective interference remains finite (the same scaling that keeps the Central Limit Theorem (CLT) noise under control for the retrieved set).
In that case, a growing fraction of patterns can be recalled in parallel with approximately comparable amplitude per pattern.}
\begin{figure}[t!]
    \centering
    \includegraphics[width=15cm]{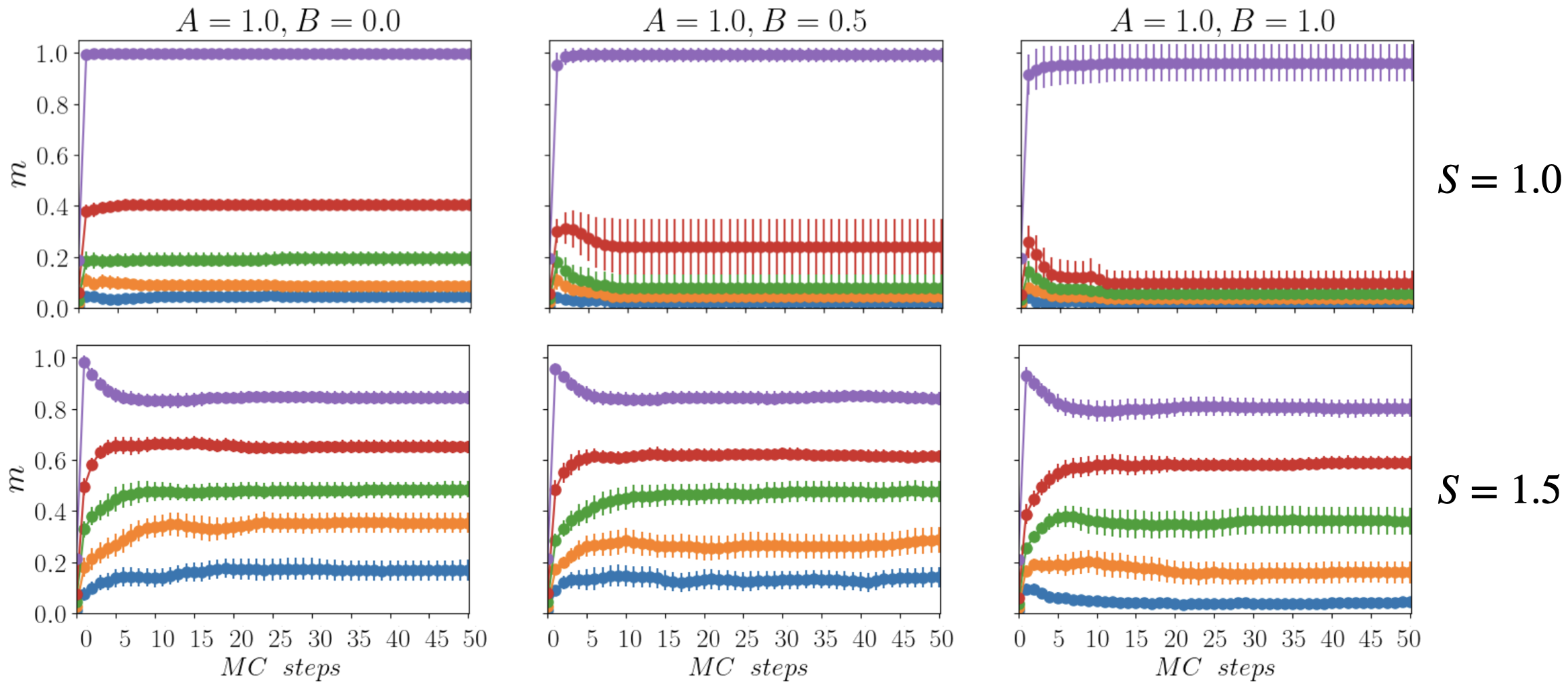}
    \caption{MCMC simulations of \eqref{eq:HamiltSspin} with \(K=5\), \(\beta=1000\), \(N=1000\), and \(a=0.6\).
    The pre-factor \(A\) multiplies the \(\xi\)-dependent term, while \(B\) multiplies the \(\eta\)-dependent term.
    \textit{Top row:} \(0\in\Omega\). Increasing \(B\) progressively suppresses multitasking until only one magnetization remains non-zero.
    \textit{Bottom row:} \(0\notin\Omega\). The same increase in \(B\) has a much weaker impact, and multiple overlaps can coexist even at larger \(B\).}
    \label{fig:comparison_terms}
\end{figure}

\subsection{The role of the \texorpdfstring{\(\eta\)}{eta}-dependent terms}
\label{subsec:eta_term}
\resub{
Dilution generically promotes multitasking through the \(\xi\)-dependent Hebbian sector. In the BEG/GS extensions, however, the additional \(\eta\)-dependent interaction introduces a competing energetic drive that can reduce (or, depending on parameters, preserve) parallel retrieval.
This competition is particularly transparent when we compare different GS models: when \(0\in\Omega\) (integer \(S\), e.g. the BEG case, \ \(S=1\)) parallel states typically require stronger dilution (smaller \(a\)), whereas when \(0\notin\Omega\) (half-integer \(S\), e.g.\ \(S=3/2\)) multitasking remains stable already at weaker dilution.
This qualitative difference is visible dynamically in Fig.~\ref{fig:comparison_terms} (last column on the right) and at the energetic level in Fig.~\ref{fig:energy_comparison}.
\medskip
\noindent
To isolate the mechanism, we compare two prototypical retrieval organizations:
\begin{itemize}
    \item the \emph{pure} retrieval state (serial recall),
    \begin{equation}
        \bm\sigma^{(s)}=\bm\xi^1\,,
        \label{eq:pure_state}
    \end{equation}
    \item the \emph{hierarchical/parallel} retrieval state \eqref{eq:parallel_state} (multitasking recall).
\end{itemize}
\noindent
It is then convenient to decompose the Hamiltonian into its \(\xi\)-dependent and \(\eta\)-dependent parts:
\begin{align}
\big|\mathcal H^{(\xi)}_{N,K,a}(\bm\sigma)\big|
=\frac{1}{2N}\sum_{i,j=1}^N\sum_{\mu=1}^K \xi_i^\mu\xi_j^\mu\,\sigma_i\sigma_j,
\qquad
\big|\mathcal H^{(\eta)}_{N,K,a}(\bm\sigma)\big|
=\frac{1}{2N}\sum_{i,j=1}^N\sum_{\mu=1}^K \eta_i^\mu\eta_j^\mu\,\sigma_i^2\sigma_j^2,
\label{eq:hamiltonian_due_pezzi}
\end{align}
so that the total energy is the sum of the two contributions.
\noindent
For the pure state one finds
\begin{align}
\big|\mathcal H^{(\xi)}_{N,K,a,S}(\bm\sigma^{(s)})\big|
=\frac{\N^2}{2},
\qquad
\big|\mathcal H^{(\eta)}_{N,K,a,S}(\bm\sigma^{(s)})\big|
=\frac{\N^2}{2},
\label{eq:num_puro}
\end{align}
whereas for the hierarchical/parallel state
\begin{align}
\big|\mathcal H^{(\xi)}_{N,K,a,S}(\bm\sigma^{(h)})\big|
=&\frac{\N^2}{2}\sum_{\nu=1}^{K}(1-a)^{2(\nu-1)},
\notag \\
\big|\mathcal H^{(\eta)}_{N,K,a,S}(\bm\sigma^{(h)})\big|
=&\frac{1}{2}\sum_{\nu=1}^{K}\Bigg(\N(1-a)^{\nu-1}-\N^2\sum_{\rho=\nu}^{K-1}(1-a)^\rho\Bigg)^2.
\label{eq:num_ger}
\end{align}
\noindent
Figure~\ref{fig:energy_comparison} clarifies the resulting competition.
The \(\xi\)-term systematically rewards the hierarchical/parallel organization, since distributing the alignment over multiple partially informative patterns increases \(\sum_\mu m_\mu^2\).
By contrast, the \(\eta\)-term depends sensitively on the available spin states: for \(S=1\) (\(0\in\Omega\)) it increasingly favors the pure state as \(a\) grows, thereby counteracting the \(\xi\)-driven tendency toward multitasking and producing a dilution threshold beyond which serial retrieval becomes energetically preferable. Instead, for \(S=1.5\) (\(0\notin\Omega\)) the same contribution remains largely compatible with multitasking, so that the hierarchical/parallel state stays energetically dominant at low temperature.\\
Building on this decomposition, we can probe the competition directly by introducing two tunable prefactors \(A\) and \(B\) multiplying, respectively, the \(\xi\)-dependent and \(\eta\)-dependent sectors in \eqref{eq:HamiltSspin}:
\begin{align}
\mathcal H^{(GS)}_{N,K,a,S}(\boldsymbol\sigma \mid \boldsymbol\xi)
&=-A\big|\mathcal H^{(\xi)}_{N,K,a,S}(\bm\sigma)\big|
-B \big|\mathcal H^{(\eta)}_{N,K,a,S}(\bm\sigma)\big|\,.
\label{eq:AB_hamiltonian}
\end{align}
This parametrization isolates how the \(\eta\)-sector reshapes retrieval: increasing \(B\) selectively enhances the term that can (for integer \(S\)) penalize multitasking.
The numerical outcomes are summarized in Fig.~\ref{fig:comparison_terms}: when \(0\in\Omega\), a sufficiently large \(B\) drives the system toward a pure state with a single condensed overlap, whereas when \(0\notin\Omega\) the same mechanism is much less effective and multiple overlaps can remain stable over a broader range of \(B\).
}

\section{Serial processing in the high load regime}
\label{sec:serial}
In this subsection we analyze these two neural networks, namely the BEG and GS models, confined to the high storage regime (namely when the number of patterns $K$ scales linearly with the system size $N$, i.e. $\delta=1$ thus $K=\alpha N$ with $\alpha>0$). With such a huge number of patterns $K$, if the dilution parameter $a$ does not scale with the volume (i.e. $a \neq a(N)$) there is solely serial retrieval and a study of these networks in such a regime has already exploited in the past via the replica trick \cite{BlumeCambelNN1,BlumeCambelNN2}. As a result, purpose of this introductory inspection is not to explore more complex tasks (w.r.t. serial pattern recognition) these networks may enjoy, but rather to confirm the heuristic picture achieved by re-obtaining the exact expression for the network's free energy by relying upon the (mathematically rigorous) Guerra's interpolation: this technique will be introduced in the next Sec. \ref{subsec:stat_mecc_app_BEG}  and, in Sec. \ref{sec:RS-BEGNN}, we use it to confirm sharply the existing replica symmetric picture.

\subsection{Blume-Emery-Griffiths model: statistical mechanics approach}
\label{subsec:stat_mecc_app_BEG}

As usual at work with neural networks whose patterns are random (hence they are approximately orthogonal for large $N$), we have to split one pattern --candidate to be retrieved-- from all the remaining ones, that --instead-- play as quenched noise against the retrieval of the chosen one. Without loss of generality (given the arbitrary of the pattern labels), we focus on the retrieval of the first pattern ($\mu=1$): as a result, we split the {\em signal term} ($\mu = 1$) from the {\em noise terms} ($\mu > 1$) in the Hamiltonian, that is, we write the partition function of the BEG network as 
\begin{align}
\mathcal{Z}_{N,K,a}(\beta\mid\boldsymbol{\xi})
= \sum_{\{\boldsymbol\sigma\}} \exp\Bigg\{
\frac{\beta N}{2}\,a\big(m_1^2+(1-a)M_1^2\big)
+ \frac{\beta N}{2}\,a \sum_{\mu>1}^K \big(m_\mu^2+(1-a)M_\mu^2\big)
\Bigg\}.
\label{eq:part_funct_BEG}
\end{align}

To handle the \textit{noise} terms ($\mu>1$), we exploit the duality of representation between (recurrent) Hebbian networks and (two-layer) Restricted Boltzmann Machines (RBMs)  \cite{Dualita1,Dualita2,Dualita3,Dualita4}, which also holds for the BEG networks as we  show  hereafter (a snapshot of the equivalence is drawn in Figure \ref{fig:NNRBMscheme}).  
\newline
By applying two independent Hubbard-Stratonovich transformations, one for the $\xi$-dependent contributions and one for the $\eta$-dependent ones, \resub{namely for $B\in\mathbb R$,
\begin{align}
\label{eq:hubbard}
\exp\!\left(\frac{B^2}{2}\right)
= \int \frac{dz}{\sqrt{2\pi}}\,e^{-z^2/2 + Bz},
\end{align}}
the partition function \eqref{eq:part_funct_BEG} can be rewritten as
\begin{align}
\mathcal{Z}_{N,K,a}(\beta\mid \boldsymbol{\xi})
&= \sum_{\{\boldsymbol\sigma\}} \int \mathcal D(z\tilde z)\,
\exp\Bigg[
\frac{\beta N}{2}\,a\big(m_1^2+(1-a)M_1^2\big)\notag\\
&\qquad+ \sqrt{\beta Na}\sum_{\mu>1}^K m_\mu z_\mu
+ \sqrt{\beta N(1-a)}\sum_{\mu>1}^K M_\mu\tilde z_\mu
\Bigg],
\label{eq:partition_split_BEG}
\end{align}
where
\begin{align}
\mathcal{D}(z\tilde z) = \prod\limits_{\mu>1}^K\dfrac{d z_\mu d \tilde z_{\mu}}{2\pi}\exp\left[-\dfrac{1}{2}\SOMMA{\mu>\resub{1}}{K} \Big(z_\mu^2 + \tilde{z}_\mu^2\Big)\right].
\end{align}
\resub{Now, replacing the definition of the order parameters in Remark \ref{rem:order},}
\textcolor{black}{we get}
\begin{align}
\mathcal{Z}_{N,K,a}(\beta\mid \boldsymbol{\xi})
&= \sum_{\{\boldsymbol\sigma\}} \int \mathcal D(z\tilde z)\,
\exp\Bigg[
\frac{\beta N}{2}\,a\big(m_1^2+(1-a)M_1^2\big)\notag\\
&\qquad+ \sqrt{\frac{\beta}{N}} \sum_{\mu>1}^K\sum_{i=1}^N \frac{\xi_i^\mu}{\sqrt{a}}\,\sigma_i z_\mu
+ \sqrt{\frac{\beta}{N}} \sum_{\mu>1}^K\sum_{i=1}^N \frac{\eta_i^\mu}{\sqrt{a(1-a)}}\,\sigma_i^2 \tilde z_\mu
\Bigg],
\end{align}
For computational convenience, in the noisy terms, we define normalized pattern variables
\begin{equation}
\tilde\xi_i^\mu := \frac{\xi_i^\mu}{\sqrt{a}}, \qquad 
\tilde\eta_i^\mu := \frac{\eta_i^\mu}{\sqrt{a(1-a)}},
\label{eq:xi_eta_tilde}
\end{equation}
so that the partition function takes the compact form
\begin{align}
\mathcal{Z}_{N,K,a}(\beta\mid \boldsymbol{\xi})
&= \sum_{\{\boldsymbol\sigma\}} \int \mathcal D(z\tilde z)\,
\exp\Bigg[
\frac{\beta N}{2}\,a\Big(m_1^2(\boldsymbol\sigma)+(1-a)M_1^2(\boldsymbol\sigma)\Big)\notag\\
&\qquad+ \sqrt{\frac{\beta}{N}} \sum_{\mu>1}^K \sum_{i=1}^N \tilde\xi_i^\mu \sigma_i z_\mu
+ \sqrt{\frac{\beta}{N}} \sum_{\mu>1}^K \sum_{i=1}^N \tilde\eta_i^\mu \sigma_i^2 \tilde z_\mu
\Bigg].
\label{eq:parition_start_BEG}
\end{align}

\begin{figure}[t!]
    \centering
    \includegraphics[width=15cm]{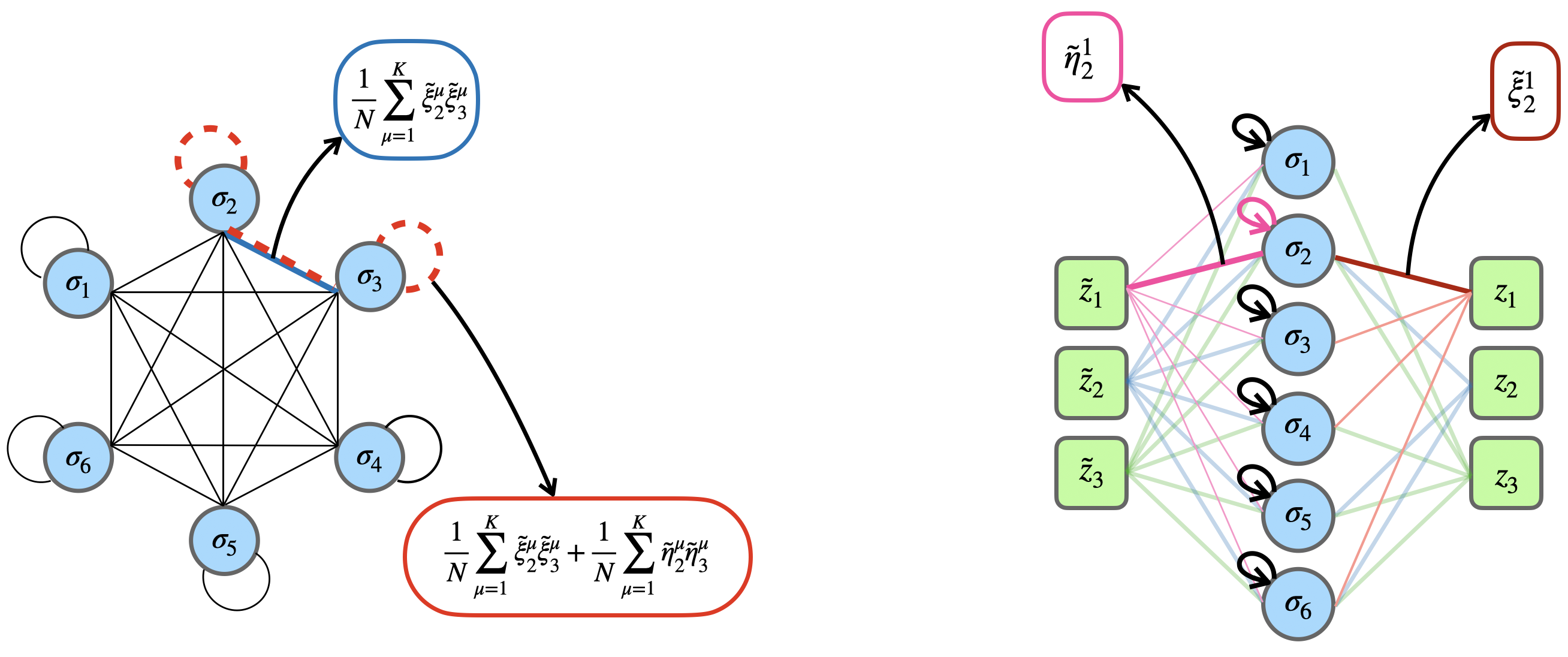}
    \caption{Schematic representation of the BEG neural network (left) and the corresponding generalized restricted Boltzmann machine (right). In the recurrent neural network on the left, for each pair of the neurons there are two contributions within their couplings: the former (highlighted in blue) accounts for the interactions with $\bm \xi$ while the latter accounts for those labeled by $\bm \eta$ (highlighted in red). Instead, in the dual RBM, each neuron $\sigma$ interacts with two different layers --of the same width $K$-- built off by Gaussian neurons $\bm z$ and $\bm \tilde z$: one layer conveys interactions mitigated by $\bm \xi$, the other accounts for those pertaining to the $\bm \eta$ as prescribed accordingly to eq. \eqref{eq:parition_start_BEG}.  Note, in both the networks, the presence of autapses on the $\bm\sigma$ layer, that is, self-couplings in statistical mechanical jargon \cite{Gallius,GalliusTwo}.}
    \label{fig:NNRBMscheme}
\end{figure}

The next step is to rely upon the universality of quenched disorder, as established in \cite{CarmonaWu} for one-party spin glasses and in \cite{Genovese} for bipartite spin glasses (such as the generalized RBM considered above) \cite{B-war2}: this universality property allows  to replace, in the noisy terms of the Hamiltonian,  the discrete-valued quenched variables $\tilde{\boldsymbol\xi}$ and $\tilde{\boldsymbol\eta}$ by standard Gaussians. In the asymptotic limit $N \to \infty$, due to the summations that are always involved over these variables when they appear (and the usage of  the CLT that automatically kicks in), the expression of the quenched statistical pressure is left unaffected by this change of variables. Yet, the reward for us is that Gaussian disorder permits the direct use of Guerra’s interpolation scheme, that heavily relies on Wick's Theorem (also known as Stein's Lemma).  

Furthermore, although $\tilde{\boldsymbol\xi}$ and $\tilde{\boldsymbol\eta}$ are constructed from the same underlying patterns, the Gaussian substitution can be performed independently for the two families. A sketch of the argument is as follows.

For the normalized disorder variables $\{\tilde\xi_i^\mu\}$ and $\{\tilde\eta_i^\mu\}$ we have
\begin{equation}
\mathbb E[\tilde\xi_i^\mu]=0,\quad \mathrm{Var}[\tilde\xi_i^\mu]=1,
\qquad
\mathbb E[\tilde\eta_i^\mu]=0,\quad \mathrm{Var}[\tilde\eta_i^\mu]=1.
\end{equation}
Moreover, by symmetry, $\mathrm{Cov}(\tilde\xi_i^\mu,\tilde\eta_i^\mu)=0$ (see Eq. \eqref{eq:cov}).
\\
Since we are not interested in single entries but only in their collective contribution through sums over $i$ and $\mu$, the Central Limit Theorem applies. In the limit $N\to\infty$, the joint disorder vector
\begin{align}
Z_i^\mu:=(\tilde\xi_i^\mu,\tilde\eta_i^\mu)\in\mathbb R^2
\end{align}
behaves as i.i.d.\ bivariate Gaussian with mean zero and covariance
\begin{align}
\mathrm{Cov}(Z_i^\mu)=\begin{pmatrix}
\mathrm{Var}[\tilde\xi_i^\mu] & \mathrm{Cov}[\tilde\xi_i^\mu,\tilde\eta_i^\mu]\\
\mathrm{Cov}[\tilde\xi_i^\mu,\tilde\eta_i^\mu] & \mathrm{Var}[\tilde\eta_i^\mu]
\end{pmatrix}
=
\begin{pmatrix}
1 & 0\\
0 & 1
\end{pmatrix}
=I_2,
\end{align}
which has independent components.  

Thus, in the thermodynamic limit, $\tilde{\boldsymbol\xi}$ and $\tilde{\boldsymbol\eta}$ can be replaced by two independent families of standard Gaussian variables. 

Thus, exploiting the universality of quenched disorder, the partition function \eqref{eq:parition_start_BEG} now depends on three distinct families of quenched random fields: the retrieved pattern $\boldsymbol{\xi}^1$ (kept discrete), and two noise contributions generated by the remaining patterns, namely $\tilde{\boldsymbol\xi}^\mu$ and $\tilde{\boldsymbol\eta}^\mu$ with $\mu>1$, which can be treated as independent Gaussian variables:
\begin{equation}
\begin{array}{lll}
     \mathcal{Z}_{N,K,a}(\beta\mid \boldsymbol{\xi}^1, \tilde{\boldsymbol\xi}, \tilde{\boldsymbol\eta})
     &= \SOMMA{\{\boldsymbol\sigma\}}{} \int \mathcal D(z,\tilde z)\,
     \exp\Bigg[\dfrac{\beta N}{2}\,a\Big(m_1^2(\boldsymbol\sigma)+(1-a)M_1^2(\boldsymbol\sigma)\Big)
     \\\\
     &\quad+ \sqrt{\dfrac{\beta}{N}} \SOMMA{\mu>1}{K} \SOMMA{i=1}{N} \tilde\xi_i^\mu \sigma_i z_\mu 
     + \sqrt{\dfrac{\beta}{N}} \SOMMA{\mu>1}{K} \SOMMA{i=1}{N} \tilde\eta_i^\mu \sigma_i^2 \tilde z_\mu \Bigg].
\end{array}
\label{eq:parition_start_BEG_1}
\end{equation}

This representation will serve as the starting point for the statistical-mechanics analysis.

\begin{remark}
 {
The integral representation allows us to connect the fully-connected BEG network to a dual representation in RBM form.  
To better understand the duality that emerges at the statistical-mechanics level between the BEG neural network and its RBM counterpart, we refer to Fig.~\ref{fig:NNRBMscheme}, which shows a toy example with six Ising spins (left) and the corresponding RBM structure (right).}  

\par\medskip
 {For clarity, let us recall the Hamiltonian of the BEG neural network:}
\begin{align}
\mathcal H^{(BEG)}_{N,K,a}(\boldsymbol{\sigma}\mid \boldsymbol{\xi})
&=-\frac{1}{2Na}\sum_{i,j=1}^N\sum_{\mu=1}^K \xi_i^\mu\xi_j^\mu\,\sigma_i\sigma_j
-\frac{1}{2Na(1-a)}\sum_{i,j=1}^N\sum_{\mu=1}^K \eta_i^\mu\eta_j^\mu\,\sigma_i^2\sigma_j^2.
\end{align}
 {Its partition function, \resub{in particular its integral representation exploiting \eqref{eq:hubbard}}, can be written as}
\begin{equation}
\label{eq:catena}
    \begin{array}{lll}
        &\mathcal{Z}_{N,K,a}(\beta\mid \boldsymbol\xi) 
        = \sum_{\{\bm \sigma\}} e^{-\beta \mathcal{H}^{\textrm{BEG}}_{N,K,a}(\boldsymbol{\sigma}|\boldsymbol{\xi})} 
        \\\\
        &= \sum_{\{\boldsymbol\sigma\}} \int D(z \tilde z)\,
     \exp\Bigg[\dfrac{\beta}{\sqrt{N}}\sum_{\mu=1}^K \sum_{i=1}^N \tilde\xi_i^\mu \sigma_i z_\mu 
     + \dfrac{\beta}{\sqrt{N}}\sum_{\mu=1}^K \sum_{i=1}^N \tilde\eta_i^\mu \sigma_i^2 \tilde z_\mu \Bigg]
     \\\\
     &= \sum_{\{\boldsymbol\sigma\}} \int D(z \tilde z)\,
     e^{-\beta \mathcal{H}^{\textrm{RBM}}_{N,K,a}(\boldsymbol{\sigma}, \bm z, \tilde{\bm z}\mid \tilde{\boldsymbol{\xi}},\tilde{\boldsymbol{\eta}})}
     =\mathcal{Z}^{\textrm{RBM}}_{N,K,a}(\beta; \tilde{\boldsymbol{\xi}}, \tilde{\boldsymbol{\eta}}),
    \end{array}
\end{equation}
\resub{where $z_\mu$ and $\tilde z_\mu$ are now supposed $\mathcal{N}(0, \beta^{-1})$, $\mu=1, \hdots, K$.}
 {The exponent in the last line of \eqref{eq:catena} defines the RBM Hamiltonian:
\begin{align}
\mathcal{H}^{\textrm{RBM}}_{N,K,a}(\boldsymbol{\sigma}, \bm z, \tilde{\bm z}\mid \tilde{\boldsymbol{\xi}},\tilde{\boldsymbol{\eta}})
= -\frac{1}{\sqrt{N}}\sum_{\mu=1}^K \sum_{i=1}^N \tilde\xi_i^\mu \sigma_i z_\mu 
     - \frac{1}{\sqrt{N}}\sum_{\mu=1}^K \sum_{i=1}^N \tilde\eta_i^\mu \sigma_i^2 \tilde z_\mu.
\end{align}
This corresponds to the cost function of a RBM with a visible layer of $N$ neurons $\{\sigma_i\}$ and two hidden layers: the first consisting of $K$ Gaussian units $\{z_\mu\}$, the second of $K$ Gaussian units $\{\tilde z_\mu\}$, both with Gaussian priors.  
In this mapping, the pattern entries $\xi_i^\mu$ in the BEG model act as the synaptic weights connecting $\sigma_i$ to $z_\mu$ in the RBM, while the $\eta_i^\mu$ provide the weights connecting $\sigma_i^2$ to $\tilde z_\mu$. }
\end{remark}

We have already introduced the order parameters that capture  the retrieval cabilities of the network, namely the Mattis magnetizations with respect to $\boldsymbol\xi$ and $\boldsymbol\eta$ (respectively, in  \eqref{eq:m_mattis_order} and \eqref{eq:M_order}): to complete the set of the required order parameters, we now need two more observables able to capture also the {\em glassy contribution} of the quenched noise to the system. These plays a role solely in the high storage regime, where the amount of patterns scales as the amount of neurons, hence, only for this section,  we introduce also the following replica overlaps.

\begin{definition}[Replica overlaps]
\label{def:ordparam}
Given two replicas $(m,l)$ of the network, the overlap $q_{m,l}$ quantifies the similarity among their neural configurations $\boldsymbol\sigma^{(m)}$ and $\boldsymbol\sigma^{(l)}$. In the \resub{same} way, the overlap $\tilde q^{m,l}_N$ measures the similarity between their {\em activities}, namely the average number of the firing neurons: 
\begin{align}
\label{eq:order_q}
q_{m,l}(\bm\sigma) := \frac{1}{N}\sum_{i=1}^N \sigma_i^{(m)}\sigma_i^{(l)}. 
\qquad 
\tilde q_{m,l}(\bm\sigma) := \frac{1}{N}\sum_{i=1}^N (\sigma_i^{(m)})^2 (\sigma_i^{(l)})^2.
\end{align}
Furthermore, $p_{m,l}$ and $\tilde p_{m,l}$ are the overlaps between replicas of the Gaussian neurons $z$ and $\tilde z$ that appear in the integral representation of the network (see Fig. \ref{fig:NNRBMscheme}) and they are defined as
\begin{align}
\label{eq:order_p}
p_{m,l}(\bm z) := \frac{1}{K}\sum_{\mu=1}^K z_\mu^{(m)} z_\mu^{(l)}, 
\qquad
\tilde p_{m,l}(\tilde{\bm{z}}) := \frac{1}{K}\sum_{\mu=1}^K \tilde z_\mu^{(m)} \tilde z_\mu^{(l)}.
\end{align}  
These are standard spin-glass order parameters, widely used in the typical analyses of neural network models \cite{Amit,Coolen}.
\newline
Note that, for a given order parameter $X_N$, we define its mean with a bar over its symbol and such a value is expected in the asymptotic limit, namely
\[
\lim_{N\to\infty}\langle X_N\rangle := \bar{X},
\]
where $\langle\cdot\rangle$ denotes the global Boltzmann average introduced in Def.~\ref{def:global}.
\end{definition}

Note that, for simplicity, we started to omit the explicit dependence of the mean values of the order parameters on their microscopic variables (i.e. $\boldsymbol\sigma$, $\boldsymbol z$, $\boldsymbol{\tilde z}$, $\boldsymbol\xi$, and $\boldsymbol\eta$), as this information does not add much to our comprehension, while we make explicit their dependence (through the global averages) on $\beta$, $N$, $K$ and $a$.

\begin{remark}[Self-overlap]
In the classical Ising case \cite{sherrington1975solvable}, where spins take only the binary values $\{\pm 1\}$, one has identically $\tilde q_{m,l}(\bm\sigma)\equiv 1$.   Here, where spins can also take the value $0$, $\tilde q^{m,l}$ carries non-trivial information because it quantifies how many neurons are behaving accordingly to the configuration represented by the pattern under retrieval. 
\newline
Moreover, note that $q_{1,1}(\bm\sigma)=\tilde q_{1,1}(\bm\sigma)$, since $(\sigma_i^{(1)})^4=(\sigma_i^{(1)})^2$ for $\sigma_i\in\{-1,0,+1\}$, yet, while this identity works here for the BEG neural network, it will no longer hold once we extend the support of the neurons in the GS generalization of the BEG neural network.
\end{remark}

\subsection{Blume-Emery-Griffiths model: replica symmetric picture}
\label{sec:RS-BEGNN}

Once the model and the associated statistical mechanical framework have been introduced, we are ready to apply Guerra’s interpolation in order to obtain an explicit expression for its quenched free energy. In these regards, in this section we make explicit usage of the next
\begin{assumption}[Replica Symmetry]
\label{RS-assumption}
In the \emph{Replica Symmetric} (RS) approximation, we assume that the variances of the order parameters vanish in the thermodynamic limit. That is, all order parameters are self-averaging around their unique mean values.   Formally, for a given order parameter $X_N$, we assume that
\begin{align}
\big\langle (X-\bar X)^2 \big\rangle_N \xrightarrow[N\to\infty]{} 0,
\end{align}
where $\bar X$ denotes the expectation of $X$.
\end{assumption}

Due to the RS Assumption \ref{RS-assumption}, it is trivial to prove that the following relations hold in the thermodynamic limit
\begin{proposition}[RS identities]
\label{prop:RS-identities}
Under Assumption~\ref{RS-assumption}, the following relations hold in the thermodynamic limit:
\begin{align}
\lim_{N\to\infty}\big\langle m_1^2 - 2\bar m\, m_1 \big\rangle_N &= -\bar m^2,
&\qquad
\lim_{N\to\infty}\big\langle M_1^2 - 2\bar M\, M_1 \big\rangle_N &= -\bar M^2,
\\[4pt]
\lim_{N\to\infty}\big\langle q_{12}p_{12} - \bar p\,q_{12} - \bar q\,p_{12} \big\rangle_N &= -\bar q\,\bar p,
&\qquad
\lim_{N\to\infty}\big\langle \tilde q_{12}\tilde p_{12} - \bar{\tilde p}\,\tilde q_{12} - \bar{\tilde q}\,\tilde p_{12} \big\rangle_N &= -\bar{\tilde q}\,\bar{\tilde p},
\\[4pt]
\lim_{N\to\infty}\big\langle p_{11} - \bar P \big\rangle_N &= -\bar P,
&\qquad
\lim_{N\to\infty}\big\langle \tilde q_{11}\tilde p_{11} - \bar{\tilde P}\,\tilde q_{11} - \bar{\tilde Q}\,\tilde p_{11} \big\rangle_N &= -\bar{\tilde Q}\,\bar{\tilde P}.
\end{align}
Here $\bar m$, $\bar M$, $\bar q$, $\bar p$, $\bar{\tilde q}$, $\bar{\tilde p}$, $\bar{\tilde Q}$, and $\bar{\tilde P}$ denote the expectation values of the corresponding order parameters $m_1$, $M_1$, $q_{12}$, $p_{12}$, $\tilde{q}_{11}$, $\tilde{p}_{11}$, $\tilde{q}_{12}$ and $\tilde{p}_{12}$.
\end{proposition}

We are now in a position to introduce the interpolation scheme à la Guerra \cite{GuerraNN}:  note that, rather than dealing directly with the BEG network coded by eq. \eqref{eq:Hamiltonian}, it is more convenient to deal with its integral representation \eqref{eq:parition_start_BEG}, which is more suitable for the interpolation framework as it shines in the next 
\begin{definition}
Let $t \in [0,1]$ be  an interpolating parameter and $A, \  B, \ C, \  E, \ F, \ G,\ H_1,\ H_2,\ \varphi,\ \psi$ constants to be set a posteriori. The interpolating partition function is defined as 
    \begin{align}
    \mathcal{Z}_{N, K, a}&(\beta;t \vert \boldsymbol{\xi}^1, \tilde{\boldsymbol\xi}, \tilde{\boldsymbol\eta}, \bm K, \bm J) = \sum_{\bm \sigma} \int \mathcal{D}(z \tilde z) \exp \left\{ t\dfrac{\beta N a}{2} m_1^2(\bm\sigma) + t\dfrac{\beta N a (1-a)}{2}  M_1^2(\bm\sigma) \right. \notag \\
    &+ \sqrt{\dfrac{\beta t}{N}} \sum_{i=1}^N \sum_{\mu >1}^K \tilde{\xi}_i^\mu \sigma_i z_\mu  +  \sqrt{\dfrac{\beta t}{N}} \sum_{i=1}^N \sum_{\mu >1}^K \tilde{\eta}_i^\mu \sigma_i^2 \tilde z_\mu + (1-t) N\Big(\psi  m_1(\bm\sigma) + \varphi M_1(\bm\sigma) \Big)\notag \\
    &+ \dfrac{(1-t)}{2} \left[ F \sum_{\mu=2}^K z_\mu^2 + G \sum_{\mu=2}^K \tilde z_\mu^2 + (H_1+H_2) \sum_{i=1}^N \sigma_i^2 \right] \notag \\
    &\left.+ \sqrt{1-t} \left[ A \sum_{i=1}^N K_i \sigma_i + B \sum_{\mu=2}^K \bar{J}_\mu z_\mu + C \sum_{\mu=2}^K \tilde J_\mu \tilde z_\mu + E \sum_{i=1}^N \tilde K_i \sigma_i^2\right]\right\}, 
    \label{eq:interp}
\end{align}
where we have denoted $\bm J=(J_\mu, \tilde J_\mu)_{\mu=1, \hdots, K}$, $\bm K=(K_i, \tilde K_i)_{i=1, \hdots, N}$ and $K_i, \tilde K_i \sim \mathcal{N}(0,1)$, $i=1, \hdots, N$,  $\bar J_\mu, \tilde J_\mu \sim \mathcal{N}(0,1)$. 

Once we have $\mathcal{Z}_{N, K, a}(\beta;t \vert \boldsymbol{\xi}^1, \tilde{\boldsymbol\xi}, \tilde{\boldsymbol\eta}, \bm K, \bm J)$, the interpolating quenched statistical pressure simply reads as
\begin{equation}\label{Ainterpolante}
\mathcal{A}_{N, K, a}(\beta;t)= \frac{1}{N}\mathbb{E} \mathcal{Z}_{N, K, a}(\beta;t \vert \boldsymbol{\xi}^1, \tilde{\boldsymbol\xi}, \tilde{\boldsymbol\eta}, \bm K, \bm J),
\end{equation}
where now $\mathbb{E}= \mathbb{E}_{\bm\xi^1}\mathbb{E}_{\tilde{\bm\xi}}\mathbb{E}_{\tilde{\bm\eta}}\mathbb{E}_{\bm K}\mathbb{E}_{\bm J}$.
\\
Furthermore $\mathcal{A}^{RS}_{a}(\alpha,\beta;t) = \lim_{N \to \infty}\mathcal{A}_{N, K, a}(\beta;t)$.
\end{definition}

\begin{remark}
    We stress that the interpolating partition function, and consequently the interpolating quenched statistical pressure, trivially recovers two important limit models in the $t \to 0$ and $t \to 1$ limits, namely 
    \begin{itemize}
        \item For $t=1$, \eqref{eq:interp} returns the corresponding integral representation of the partition function of the original model \eqref{eq:Hamiltonian}; 
        \item For $t=0$, \eqref{eq:interp} reduces to a set of effective one-body models whose quenched statistical pressure is trivial to compute. 
    \end{itemize}
The particular values of these two limits let us understand the interpolation criterion: while we are able to compute everything at $t=0$ but we want everything at $t=1$, we can apply the Fundamental Theorem of Calculus (FTC) on the interpolating quenched statistical pressure to connect these two extrema, namely 
    \begin{align}
        \mathcal{A}^{RS}_{a}(\alpha,\beta)=\mathcal{A}^{RS}_{a}(\alpha,\beta;t=1)&=\lim\limits_{N\to\infty}\mathcal{A}_{N, K, a}(\beta;t=1)\notag \\
        &=  \lim\limits_{N\to\infty}\left(\mathcal{A}_{N, K, a}(\beta;t=0)+ \int_0^1 d_t \mathcal{A}_{N, K, a}(\beta;t)\Big\vert_{t=s} ds\right).
        \label{eq:FTC}
    \end{align}
\end{remark}
Note that, in the following, in order to carry all the computations involved in the above scheme, it will be mandatory the usage of Assumption \ref{RS-assumption} regarding the self-averaging behavior of the order parameters in the thermodynamic limit: the whole results in the main theorem of this Section that is
\begin{theorem}
\label{th:BEG_RS_final_ret_seriale}
The thermodynamic limit of the quenched statistical pressure of the Blume-Emery-Griffiths neural network coded by the Hamiltonian \eqref{eq:Hamiltonian}, under the replica symmetric assumption, reads explicitly as 
\begin{align}
    \mathcal{A}^{RS}_{\alpha, a}(\beta)=& \dfrac{\alpha }{2} \dfrac{\beta \qb }{\left(1-\beta (\Qb - \qb)\right)} - \dfrac{\alpha}{2} \log\left(1-\beta(\Qb - \qb)\right) + \dfrac{\alpha}{2}\dfrac{\beta \bar{\qt}}{\left(1-\beta (\Qb - \bar{\qt})\right)}\notag \\
    &- \dfrac{\alpha}{2} \log \left(1-\beta(\Qb - \bar{\qt})\right)+ \mathbb{E} \log \left\{ 1+2 \exp \left[ g_1(\xi^1, \bar{J}) \right] \cosh \left[ g_2(\xi^1, J) \right]\right\} \label{eq:ARS} \\
    &-\dfrac{\beta}{2} a \mb_1^2 - \dfrac{\beta }{2} a(1-a)\Mb_1^2-\dfrac{\beta \alpha }{2} \Qb\left( \Pb +\bar{\Pt}\right) +\dfrac{\beta \alpha }{2} \left(\pb\qb +\bar{\pt}\bar{\qt} \right),\notag
\end{align}
where $\mathbb{E}=\mathbb{E}_{\xi^1}\mathbb{E}_J \mathbb{E}_{\tilde J}$ and
\begin{align}
    g_1(\xi^1, \bar{J}):=& \beta \Mb [(\xi^1)^2-a] + \sqrt{\beta \alpha \bar{\pt}} \tilde J + \dfrac{\beta \alpha}{2}\left( \Pb +\bar{\Pt} \right) - \dfrac{\beta \alpha}{2}\left(\pb  + \bar{\pt}\right), \\
    g_2(\xi^1, J):=& \beta \mb \xi^1 + \sqrt{\beta \alpha \pb} J,
\end{align}
and whose order parameters fulfill the following self-consistency equations
\begin{equation}
    \mb_1= \mathbb{E} \left[ \dfrac{\tanh \left[ g_2(\xi^1, J) \right]}{1+\frac{1}{2} e^{-g_1(\xi^1, \bar{J})} \cosh^{-1} \left[ g_2(\xi^1, J) \right] } \left(\dfrac{\xi^1}{a} \right)\right],
\end{equation}
\begin{equation}
    \Mb_1=  \mathbb{E} \left[ \dfrac{1}{1+\frac{1}{2} e^{-g_1(\xi^1, \bar{J})} \cosh^{-1} \left[ g_2(\xi^1, J) \right]  } \left(\dfrac{(\xi^1)^2-a}{a(1-a)}\right)\right], 
\end{equation}
\begin{equation}
    \qb= \mathbb{E} \left[ \dfrac{\tanh^2 \left[ g_2(\xi^1, J) \right] }{(1+\frac{1}{2} e^{-g_1(\xi^1, \bar{J})} \cosh^{-1} \left[ g_2(\xi^1, J) \right] )^2 } \right],
\end{equation}
\begin{equation}
    \bar{\qt}= \mathbb{E} \left[ \dfrac{1}{(1+\frac{1}{2} e^{-g_1(\xi^1, \bar{J})} \cosh^{-1} \left[ g_2(\xi^1, J) \right])^2  } \right],
\end{equation}
\begin{equation}
    \Qb = \mathbb{E} \left[  \dfrac{1}{1+\frac{1}{2} e^{-g_1(\xi^1, \bar{J})} \cosh^{-1} \left[ g_2(\xi^1, J) \right] } \right],
\end{equation}
\begin{equation}
    \pb = \dfrac{\beta \qb }{\left( 1-\beta (\Qb-\qb)\right)^2},\qquad \qquad 
    \bar{\pt}=\dfrac{\beta \qt}{\left( 1-\beta (\Qb-\bar{\qt})\right)^2},
\end{equation}
\begin{equation}
    \Pb + \bar{\Pt} = \pb + \bar{\pt} + \dfrac{1}{1-\beta(\Qb-\qb)} + \dfrac{1}{\left( 1-\beta(\Qb-\bar{\qt})\right)}.
\end{equation}
\end{theorem}

To prove the theorem, we need to premise the following 
\begin{lemma}
    The derivative with respect to $t$ of the interpolating quenched statistical pressure in thermodynamic limit and under the replica symmetric assumption is 
    \begin{align}
    d_t \mathcal{A}^{RS}_{\alpha, a}(\beta;t) = -\dfrac{\beta}{2}a \mb^2 - \dfrac{\beta }{2}a(1-a) \Mb^2-\dfrac{\beta \alpha }{2} \left( \Pb \Qb-\pb \qb\right) -\dfrac{\beta \alpha }{2} \left(\bar{\Pt} \Qb-\bar{\pt}\bar{\qt} \right).
    \label{eq:dtA}
\end{align}
\label{lemma1}
\end{lemma}

\begin{proof}
By direct derivation of the interpolating quenched statistical pressure w.r.t. the interpolating parameter $t$, we get 
\begin{align}
    &d_t \mathcal{A}_{N, K, a}(\beta;t) = \dfrac{\beta a}{2} \l (m_1)^2  -\dfrac{2\psi}{\beta a} m_1\r + \dfrac{\beta a(1-a)}{2} \l (M_1)^2 -\dfrac{2\varphi}{\beta a(1-a)}M_1\r \label{eq:dtAN}
    \\
    &- \dfrac{\beta K}{2 N}\Big(\l q_{12}p_{12}\r -\dfrac{A^2 N}{K\beta} \l q_{12}\r -\dfrac{B^2}{\beta}\l p_{12}\r\Big) + \dfrac{\beta K}{2 N}\Big(\l q_{11}p_{11}\r -\dfrac{(A^2+H_1) N}{K\beta} \l q_{11}\r -\dfrac{B^2+F}{\beta}\l p_{11}\r\Big)   \notag 
    \\
    & -  \dfrac{\beta K}{2 N}\Big(\l \tilde{q}_{12}\tilde{p}_{12}\r-\dfrac{E^2 N}{K\beta} \l q_{12}\r -\dfrac{C^2}{\beta}\l p_{12}\r\Big)\notag
    \\
    &+   \dfrac{\beta K}{2 N}\Big(\l q_{11}\tilde{p}_{11}\r -\dfrac{(E^2+H_2) N}{K\beta} \l q_{11}\r -\dfrac{C^2+G}{\beta}\l \tilde p_{11}\r\Big) \notag .
\end{align}
Since the computations are standard but cumbersome, for the sake of clearness we report them in Appendix \ref{app:conti}.
\\
Thanks to the RS assumption and the relations presented in Preposition~\ref{prop:RS-identities}, with some algebra on \eqref{eq:dtAN}, putting
\begin{align}
    &\psi = \beta a\bar{m}, \qquad \varphi=\beta a(1-a)\bar{M}, \qquad A^2= \beta \alpha \pb, \qquad B^2= \beta \qb \qquad C^2=\beta \bar{\qt}, \qquad E^2 = \beta \alpha \bar{\pt} ,\notag\\
    & F=\beta(\Qb - \qb), \qquad G= \beta(\Qb-  \bar{\qt}), \qquad H_1=\beta \alpha(\Pb-\pb),\qquad H_2=\beta \alpha(\bar{\Pt}-\bar{\pt}),
    \label{eq:constRS}
\end{align}
we reach the thesis.
\end{proof}

\begin{lemma}
    The expression of the interpolating quenched statistical pressure, at $t=0$ and in the thermodynamic limit, reads as 
    \begin{align}
    \mathcal{A}^{RS}_{\alpha, a}(\beta;t=0)=& \dfrac{\alpha }{2} \dfrac{B^2 }{1-F} - \dfrac{\alpha}{2} \log(1-F) + \dfrac{\alpha}{2}\dfrac{C^2}{1-G} - \dfrac{\alpha}{2} \log (1-G) \notag \\
    &+ \mathbb{E} \log \left[ 1+2 \exp \left(\dfrac{\varphi}{a(1-a)} \eta^1 + E \tilde J + \dfrac{H_1+H_2}{2}\right)\cosh \left( \dfrac{\psi}{a} \xi^1 + A J\right)\right]
    \label{eq:ob}
\end{align}
where $A, \ B,\ C, \ E, \ F, \ G, \ H_1,\ H_2,\ \psi, \varphi$ have the values in \eqref{eq:constRS}.
\label{lemma2}
\end{lemma}

Again, for the sake of clearness, the proof of Lemma \ref{lemma2} is reported in Appendix \ref{app:ob}.

\begin{proof}
    (Of Theorem \ref{th:BEG_RS_final_ret_seriale}) The proof of the theorem is the application of the FTC on the interpolating quenched statistical pressure \eqref{eq:FTC}. Therefore, putting together \eqref{eq:ob} and \eqref{eq:dtA} we reach the thesis. 
    \newline
    Finally, to recover the expression of the self-consistency equations for the order parameters, we need to extremize the quenched statistical pressure with respect to them that is a trivial, fairly standard, exercise. 
\end{proof}

Note that the above picture, particularly the expression \eqref{eq:ARS} for the quenched statistical pressure (respectively in the high ($\alpha>0$) and low ($\alpha=0$) load regimes) coincide sharply with those previously achieved via the replica trick \cite{BlumeCambelNN2}.

\subsection{Ghatak-Sherrington model: statistical mechanics approach}
We now turn to the investigation of the Ghatak-Sherrington (GS) neural network model, whose Hamiltonian is \eqref{eq:HamiltSspin}. As usual, the first steps lie in splitting the signal term (i.e. $\mu=1$ with no loss of generality) from the quenched noise term (i.e. all the contributions for $\mu \neq 1$) and to rely upon the integral representation of the partition function, namely to obtain the generalized restricted Boltzmann machine, dual to the GS model. In these regards, we state directly the following

\begin{proposition}
The partition function of the (GS) neural network model, defined by the Hamiltonian  \eqref{eq:HamiltSspin}, can be written as
    \begin{align}
        \mathcal{Z}_{N, K, a, S}(\beta\mid \boldsymbol{\xi}^1, \tilde{\boldsymbol\xi}, \tilde{\boldsymbol\eta})&= \sum_{\{\boldsymbol\sigma\}} \int \mathcal D(z\tilde z)\,\exp\Bigg[\frac{\beta N}{2}\,\Big(\N m_1^2(\boldsymbol\sigma)+\Nn M_1^2(\boldsymbol\sigma)\Big)\notag
        \\
        &\qquad+ \sqrt{\frac{\beta}{N}} \sum_{\mu>1}^K \sum_{i=1}^N \tilde\xi_i^\mu \sigma_i z_\mu+ \sqrt{\frac{\beta}{N}} \sum_{\mu>1}^K \sum_{i=1}^N \tilde\eta_i^\mu \sigma_i^2 \tilde z_\mu\Bigg].
\label{eq:parition_start_GS}
\end{align}
where $m_1$ and $M_1$ deal with the signal and the remaining terms define a generalized restricted Boltzmann machine equipped with a visible layer $\sigma$ --built off by the original neurons-- and two hidden layers $z$ and $\tilde{z}$ real-valued neurons  equipped with a Gaussian prior, i.e. the regularizator that lies in $\mathcal D(z\tilde z)$. Moreover, as usual, for $\mu>1$ we have introduced  the normalized weights
\begin{equation}
    \tilde\xi_i^\mu = \dfrac{\xi_i^\mu}{\sqrt{\N}}, \qquad \tilde\eta_i^\mu = \dfrac{\eta_i^\mu}{\sqrt{\Nn}},
\end{equation}
because, with the same universality argument exploited in the previous Section,  these can be considered as i.i.d. Normal distributed variables in the thermodynamic limit.
\end{proposition} 
All the basic definitions required for a statistical mechanical treatment of the model (e.g. partition function, quenched statistical pressure, Boltzmann and global averages, etc.) can be trivially imported from the BEG case (see Sec.~\ref{subsec:stat_mecc_app_BEG}).

\begin{remark}[Difference between GS and BEG models]
At difference w.r.t. the BEG neural network, in the GS model the identity $\boldsymbol\sigma^2=\boldsymbol\sigma^4$ does no longer hold.  
Therefore, while we preserve the order parameters  defined  in Def.~\ref{def:ordparam}, now $q_{11}\neq\tilde q_{11}$.    
\end{remark}

\subsection{Ghatak-Sherrington model: replica symmetric picture}
\label{ssec:serialSSpin}
Now, mirroring the investigation performed for the BEG model, in the thermodynamic limit  and under the replica symmetric assumption, as stated in Assumption~\ref{RS-assumption}, we recover the expression of the quenched statistical pressure of the GS model. 

Also from the methodological viewpoint, we will apply the same methodology (namely Guerra's interpolation exploiting FTC) as in the Blume-Emery-Griffiths case. 
\newline
Let us start from the definition of the interpolating quenched statistical pressure: 

\begin{definition}
    Let $t \in [0,1]$ be an interpolating parameter. The interpolating quenched statistical pressure $\mathcal{A}_{N, K, a, S}(\beta;t)$ is defined as 
    \begin{equation}\label{AinterpolanteSspin}
\mathcal{A}_{N, K, a, S}(\beta;t)= \frac{1}{N}\mathbb{E} \mathcal{Z}_{N, K, a, S}( \beta;t\vert \boldsymbol{\xi}^1, \tilde{\boldsymbol\xi}, \tilde{\boldsymbol\eta}, \bm J, \bm K),
\end{equation}
where $\mathcal{Z}_{N, K, a, S}( \beta;t\vert \boldsymbol{\xi}^1, \tilde{\boldsymbol\xi}, \tilde{\boldsymbol\eta}, \bm J, \bm K)$ is the interpolating partition function defined as 
\begin{equation}
    \begin{array}{lll}
    \mathcal{Z}_{N, K, a, S}( \beta;t\vert \boldsymbol{\xi}^1, \tilde{\boldsymbol\xi}, \tilde{\boldsymbol\eta}, \bm J, \bm K) =& \sum_{\bm \sigma} \int \mathcal{D}(z \tilde z) \exp \left\{ \dfrac{\beta t}{2} N \N m_1^2 + \dfrac{\beta t}{2} N \Nn M_1^2 + \sqrt{\dfrac{\beta t}{N}} \SOMMA{i=1}{N} \SOMMA{\mu>1}{K} \tilde \xi_i^\mu \sigma_i z_\mu \right. 
    \\\\
    & +  \sqrt{\dfrac{\beta t}{N}} \SOMMA{i=1}{N}\SOMMA{\mu>1}{K} \tilde \eta_i^\mu \sigma_i^2 \tilde z_\mu + (1-t) \left(\psi N m_1 + \varphi N M_1 \right) 
    \\\\
    &+ \dfrac{(1-t)}{2} \left[ F \SOMMA{\mu>1}{K} z_\mu^2 + G \SOMMA{\mu>1}{K} \tilde z_\mu^2 + H \SOMMA{i=1}{N} \sigma_i^2 + L \SOMMA{i=1}{N} \sigma_i^4\right] 
    \\\\
    &\left.+ \sqrt{1-t} \left[ A \SOMMA{i=1}{N} K_i \sigma_i + B \SOMMA{\mu>1}{K} \bar{J}_\mu z_\mu + C \SOMMA{\mu>1}{K} \tilde J_\mu \tilde z_\mu + E \SOMMA{i=1}{N} \tilde K_i \sigma_i^2\right]\right\}. 
\end{array}
\label{eq:interpSSpin}
\end{equation}

Note that $A, B, C, E, F, G, H, L, \psi, \varphi$ are constants to be set a posteriori (in order to force the model to respect the RS criterion), while the quenched variables $\bm K=(K_i, \tilde K_i)_{i=1, \hdots, N}$, $\bm J=(J_i, \tilde J_i)_{i=1, \hdots, N}$ are all standard Gaussians. 
\newline
Furthermore, the quenched statistical pressure in the thermodynamic limit is obtained as 
$$\mathcal{A}_{a,S}(\alpha,\beta)= \lim_{N \to \infty}\mathcal{A}_{N, K, a, S}(\beta;t=1).$$
\end{definition}

Thanks to the previous definition, it is possible to state the following
\begin{theorem}
The thermodynamic limit of the quenched statistical pressure of the Ghatak-Sherrington neural network model, defined by the Hamiltonian \eqref{eq:HamiltSspin}, under the replica symmetric assumption, reads explicitly as 
\begin{align}
    &\mathcal{A}^{RS}_{a, S}(\alpha,\beta)=\dfrac{\alpha }{2} \dfrac{\beta \qb }{\left(1-\beta(\Qb - \qb)\right)} - \dfrac{\alpha}{2} \log\left(1-\beta(\Qb - \qb)\right) + \dfrac{\alpha}{2}\dfrac{\beta \bar{\qt}}{\left(1-\beta(\bar{\tilde Q} - \bar{\qt})\right)}\notag \\
    &+ \mathbb{E}_{\xi^1, J, \tilde J} \Bigg\{\log \Bigg[\SOMMA{k=0}{2S} \exp \left( g_2(\xi^1,J) s(k) + g_1(\xi^1, \tilde J) s^2(k)+ \dfrac{\beta \alpha}{2}(\bar{\Pt} - \bar{\pt}) s^4(k)\right)\Bigg]\Bigg\}\notag \\
    &- \dfrac{\alpha}{2} \log \left(1-\beta(\bar{\tilde Q} - \bar{\qt})\right)-\dfrac{\beta}{2}\N \mb_1^2 - \dfrac{\beta }{2}\Nn \Mb_1^2-\dfrac{\beta \alpha }{2} \left( \Pb \Qb-\pb \qb\right) -\dfrac{\beta \alpha }{2} \left(\bar{\Pt} \bar{\tilde Q}-\bar{\pt}\bar{\qt} \right), 
\end{align}
where $\mathbb{E}= \mathbb{E}_{\xi^1} \mathbb{E}_J \mathbb{E}_{\tilde J}$, $s(k)=-1+\frac{k}{S}$ and 
\begin{align}
    g_1(\xi^1, \tilde J):=& \beta \Mb [(\xi^1)^2-\N] + \sqrt{\beta \alpha \pt} \tilde J + \dfrac{\beta \alpha}{2}\left( \Pb - \pb\right), \\
    g_2(\xi^1, J):=& \beta \mb \xi^1 + \sqrt{\beta \alpha \pb} J,
\end{align}
and where the order parameters must fulfill the following self-consistency equations
\begin{align}
    &\bar m =  \mathbb{E} \left[ \dfrac{\SOMMA{k=0}{2S} \mathcal{G}(k, \xi^1, J, \tilde{J}) s(k)}{\SOMMA{k=0}{2S} \mathcal{G}(k, \xi^1, J, \tilde{J})}   \left(\dfrac{\xi^1}{\N}\right)\right], 
    \\
    &\bar M = \mathbb{E} \left[ \dfrac{\SOMMA{k=0}{2S} \mathcal{G}(k, \xi^1, J, \tilde{J})s^2(k)}{\SOMMA{k=0}{2S}\mathcal{G}(k, \xi^1, J, \tilde{J})} \left(\dfrac{(\xi^1)^2}{\Nn}\right)\right] -\dfrac{\N}{\Nn}\mathbb{E} \left[ \dfrac{\SOMMA{k=0}{2S} \mathcal{G}(k, \xi^1, J, \tilde{J})s^2(k)}{\SOMMA{k=0}{2S}\mathcal{G}(k, \xi^1, J, \tilde{J})} \right], 
    \\
    &\bar q =\mathbb{E} \left[ \dfrac{\SOMMA{k=0}{2S} \mathcal{G}(k, \xi^1, J, \tilde{J})s(k)}{\SOMMA{k=0}{2S} \mathcal{G}(k, \xi^1, J, \tilde{J})} \right]^2,\quad\quad\bar{\qt} =\mathbb{E} \left[ \dfrac{\SOMMA{k=0}{2S} \mathcal{G}(k, \xi^1, J, \tilde{J})s^2(k)}{\SOMMA{k=0}{2S} \mathcal{G}(k, \xi^1, J, \tilde{J})}\right]^2, 
    \\
    &\Qb = \mathbb{E} \left[ \dfrac{\SOMMA{k=0}{2S} \mathcal{G}(k, \xi^1, J, \tilde{J}) s^2(k)}{\SOMMA{k=0}{2S}\mathcal{G}(k, \xi^1, J, \tilde{J})} \right], \quad\quad\bar{\tilde Q}=  \mathbb{E} \left[ \dfrac{\SOMMA{k=0}{2S} \mathcal{G}(k, \xi^1, J, \tilde{J})s^4(k)}{\SOMMA{k=0}{2S} \mathcal{G}(k, \xi^1, J, \tilde{J})} \right], 
    \\
    &\bar p = \dfrac{\beta \qb}{(1-\beta(\bar Q - \qb))^2} \qquad \bar{\tilde p} = \dfrac{\beta \bar{\qt} }{(1-\beta(\bar{\tilde Q} - \bar{\qt}))^2}, 
    \\
    &\bar P = \pb + \dfrac{1}{1-\beta(\Qb-\qb)} \qquad 
    \tilde P = \tilde p + \dfrac{1}{1-\beta(\bar{\tilde Q} - \bar{\qt})}, 
\end{align}
where we set
\begin{equation}
    \mathcal{G}(k, \xi^1, J, \tilde{J})=  \exp \left(g_2(\xi^1,J) s(k) + g_1(\xi^1, \tilde J) s^2(k)+ \dfrac{\beta \alpha}{2}(\bar{\Pt} - \bar{\pt}) s^4(k)\right).
\end{equation}
\end{theorem}

\begin{proof}
The proof is a straightforward generalization of the corresponding one provided for the previous BEG model hence we omit it for the sake of simplicity and we \resub{refer} the reader to the proof given for Th.~\ref{th:BEG_RS_final_ret_seriale}.
\end{proof}

We stress that, for $S=1$, the free energy of the GS neural network model naturally recovers that of the BEG neural network analyzed in Sec. \ref{sec:RS-BEGNN} while for $S=1/2$ it reduces to the standard  Amit-Gutfreund-Sompolinsky (AGS) expression of the free energy related to the standard Hopfield neural network \cite{Amit,Hopfield}.

So far, we have recovered --and minimally generalized\footnote{\resub{Strictly speaking, although the GS neural network model is a rather straightforward generalization of the BEG model, augmented with graded-response neurons, the explicit form of its free energy, together with the corresponding self-consistency equations for its order parameters, has, to the best of our knowledge, not yet appeared in the Literature.}}-- results on neural networks models (driven by {\em inverse freezing} research in spin glasses), that were already known in the Literature but this is not the end of the present analysis.
\newline
Yet, research is far from being over and, indeed, in the next section  we shall push the analysis further and investigate a completely unexplored working  regime for these models: the \emph{parallel retrieval} regime. 
\newline
As anticipated in the Introduction, for these multi-tasking capabilities to emerge,  the network has to operate respecting the following tradeoff between storage and dilution: if confined to the low storage (that is -at most- for $K\propto \ln N$), in the $a \to 0$ limit, there is enough dilution to let the network switch from serial to parallel processing, while -at work with the medium storage  (that is for $K \propto N^{\delta}$ with $\delta<1$) we can still have multitasking capabilities solely if we force the network to be extremely diluted (that is $a=a(N)$ and, in particular, $a \propto N^{-\delta}$).
\newline
For the sake of completeness, we stress that it should be possible to force these neural network to work under the massively parallel retrieval in the high storage regime --namely in an operational model where all the patterns are recalled at once and their amount is  an extensive fraction of $N$ (that is, $\delta=1$ and $\alpha >0$)-- but this requires the underlying network to be equipped with  finite connectivity: unfortunately, this request makes  the techniques for its statistical mechanical treatment to be completely different w.r.t. the interpolation so far exploited. This is due to the fact that, in this topological regime, the network can not admit a giant component any longer, rather it is fragmented in an extensive amount of small cliques dense of short loops: see \cite{Coolen1,Coolen2,EPL-NOI-12}. We remind the inspection of this regime to future explorations of these neural networks.

\section{Parallel processing in the low and medium load regimes}
\label{sec:lowloadS}
Once decided that we will not explore the finite connectivity regime, hereafter we will explore the parallel retrieval of these neural networks confined to the low and medium storage only. \resub{F}rom this Section until the end of the paper we shall work under the assumption that the number of stored patterns $K$ always grows sub-linearly with the network size $N$, namely
\begin{align}
K \propto N^{\delta}, \qquad \delta<1 \to \alpha=0,
\end{align}
\resub{where we stress that $\alpha$ is the load of the network defined in \eqref{eq:load}. Moreover, we recall that $a$ is the parameter linked to the probability distribution of the patterns (introduced for BEG in \eqref{eq:dist_xi} and for GS model in \eqref{eq:probdistrS}) and $\beta$ is the inverse temperature of the model, as reported schematically in Def. \eqref{Def:ControlParameters}. As introduced in Sec.~\ref{subsec:serial_vs_parallel_overview}, Hebbian-like networks operating away from saturation can sustain \emph{parallel retrieval}, provided that the stored patterns are sufficiently diluted. In this situation the network can raise several Mattis magnetizations at once: it leverages the reduced information content of each pattern to extract as much signal as possible by combining partial retrievals.}

\begin{remark}[Multi-tasking scenarios]
Away from the high-storage regime, we shall distinguish two cases:
\begin{enumerate}
\item \textbf{Low storage and finite dilution}  
The number of stored patterns grows at most logarithmically with the volume of the network and the dilution level remain finite as $N\to\infty$, i.e. $K \propto \ln N$ while $a \neq a(N)$.

\item \textbf{Medium storage and extreme dilution:}  
The number of patterns diverges sub-linearly with $N$, that is
\[
\alpha := \lim_{N\to\infty}\frac{K}{N^{\delta}}, \qquad \delta<1,
\]
but  the dilution in the patterns now must  also depend on $N$, so to vanish as
\[
a \propto N^{-\gamma}, \qquad \gamma\in(0,1).
\]
From a graph theory perspective, this regime forces the underlying topology to be  \emph{extremely diluted} and, anticipating what we will find soon, in order for this network to perform parallel retrieval, we must have $\delta=\gamma$, namely the way we store patterns in the network, not surprisingly, must be highly correlated to the way we diluted these patterns before supplying them to the network.
\end{enumerate}
\end{remark}

\begin{remark}
We stress that, confined away from the high storage regime (namely for the rest of the paper, where $\alpha=0$), spin-glass effects are negligible: the spin glass order parameters $q$ and $p$ are no longer needed  and the order parameters $m_\mu$ and $M_\mu$ become truly self-averaging. That is, in the thermodynamic limit there are no fluctuations around their mean values $\bar m_\mu$ and $\bar M_\mu$, hence
\begin{align}
\lim_{N\to\infty}\mathbb{P}_N(m_\mu)=\delta(m_\mu-\bar m_\mu), 
\qquad
\lim_{N\to\infty}\mathbb{P}_N(M_\mu)=\delta(M_\mu-\bar M_\mu),
\end{align}
thus the replica symmetric picture holds exactly now. 
\end{remark}

\subsection{Parallel processing in the low load regime}
The first scenario we analyze is the simplest available, that is,  when both the number of stored patterns $K$ and the dilution level $a$ stay finite in the thermodynamic limit, namely 
\begin{align}
K \in \mathbb{R}^+,\quad K>2, \qquad a\in(0,1].
\end{align}
Mirroring the  exposition of the high storage regime in Sec. \ref{sec:serial}, we will analyze BEG neural network first and, afterwards, its generalization provided by the GS model, the whole under the umbrella of Guerra's interpolation techniques.

\subsubsection{Blume-Emery-Griffiths model}
\label{sec:BEGlowload}

In order to apply Guerra's interpolation, following the path exploited in Sec. \ref{sec:serial}, we need first to introduce an interpolating quenched statistical pressure whose extrema are meaningful for the present treatment (that is one side of the interpolation recovers the true statistical pressure of the model and the other side reduces to a solvable model) and then to apply  the Fundamental Theorem of Calculus (FTC) to connect these extrema: this way we recover the expression of the quenched statistical pressure of the original model by computing, rather than directly everything at $t=1$, everything at $t=0$  adding also the contribution  of the statistical pressure derivative with respect to $t$ which is, in general, simpler to afford. 

\begin{definition}
    Let the interpolating parameter  be $t \in [0,1]$ and let $\psi_\mu, \ \varphi_\mu$, $\mu=1, \hdots, K$  be constants whose specific expression will be set a posteriori to ensure self-averaging of the order parameters. The definition of the interpolating quenched statistical pressure at finite volume is
\begin{align}\label{InterpoParBEG}
    \mathcal{A}_{N, K, a}(\beta; t)=&\dfrac{1}{N} \mathbb{E} \log \mathcal{Z}_{N, K, a}(\beta;t \vert \bm \xi)  \notag \\
    =&\dfrac{1}{N} \mathbb{E}\log \sum_{\{\bm \sigma\}} \exp \left[\dfrac{\beta N t}{2}a \sum_{\mu=1}^K (m_\mu)^2 + \dfrac{\beta N t}{2}a(1-a) \sum_{\mu=1}^K (M_\mu)^2 \right. \notag \\ 
    &\left.+ (1-t) N\sum_{\mu=1}^K \left( \psi_\mu m_\mu + \varphi_\mu M_\mu\right)\right].
\end{align}
\end{definition}
By relying upon the FDT on \eqref{InterpoParBEG}, as stated, immediately we reach the main theorem of this Section that reads 
\begin{theorem}
\label{thm:BEGnna0}
Under the parallel retrieval ansatz, in the thermodynamic limit $N\to \infty$, in the low-storage of patterns and for finite pattern's dilution values, the expression of the quenched statistical pressure of the Blume-Emery-Griffiths neural network model, whose Hamiltonian is \eqref{eq:HamiltSspin}, reads in terms of its order and control parameters as 
    \begin{align}
    \label{eq:AlowBEG}
    \mathcal{A}_{K,a}(\beta) =&  \mathbb{E}_{\bm\xi} \log \left\{ 1+2 e^{\beta G_1(\bm\xi,\bm{M})} \cosh \left[\beta G_2(\bm\xi,\bm{m}) \right]\right\}-\dfrac{\beta}{2} a\sum_\mu \Bigg(\mb_\mu^2 +(1-a)\bar{M}_\mu^2\Bigg),
\end{align}
where
\begin{equation}
    G_1(\bm\xi,\bm{M})=  \sum_{\mu=1}^K  [(\xi^\mu)^2-a] \bar M_\mu, \qquad G_2(\bm\xi,\bm{m})= \sum_{\mu=1}^K\xi^\mu \bar m_\mu,
\end{equation}
and  where the order parameters must fulfill the following self-consistency equations, for $\mu=1, \hdots, K$
\begin{equation}
    \begin{array}{lll}
        \mb_\mu=& \mathbb{E}_{\bm\xi} \left[ \dfrac{ \tanh \left[\beta G_2(\bm\xi,\mb) \right]}{ 1+\frac{1}{2} e^{-\beta G_1(\bm\xi,\bm{M})} \cosh^{-1} \left[\beta G_2(\bm\xi,\bm{m})\right]}\left(\dfrac{\xi^\mu}{a}\right) \right], 
        \\\\
        \bar{M}_\mu=& \mathbb{E}_{\bm\xi}  \left[ \dfrac{1}{ 1+\frac{1}{2} e^{-\beta G_1(\bm\xi,\bm{M})} \cosh^{-1} \left[\beta G_2(\bm\xi,\bm{m})\right]}\left(\dfrac{(\xi^\mu)^2-a}{a(1-a)}\right)\right].
    \end{array}
    \label{eq:selfMT}
\end{equation}
\end{theorem}

\begin{remark}
    We stress that the case under consideration in this Section can not be obtained by directly setting $\alpha=0$ in Sec. \ref{sec:serial} (as standard in neural networks): indeed, this discrepancy raises because the starting assumptions are  different (i.e. serial vs the parallel retrieval ansatz). By requiring $\alpha=0$ in Eq. \eqref{eq:ARS} we get \eqref{eq:AlowBEG} with $K=1$.
\end{remark}

Again the proof of Theorem \ref{thm:2} is a trivial application of the FTC \eqref{eq:FTC} that we achieve, mirroring the previous Section, with the same premises  through these following Lemmas
\begin{lemma}
    The derivative with respect to $t$ of the quenched statistical pressure of the BEG model --in the low storage regime and under the parallel retrieval ansatz--   reads, at finite size $N$, as   
    \begin{align}
        d_t \mathcal{A}_{N, K, a}(\beta;t) = \dfrac{\beta}{2} a \sum_{\mu=1}^K \l (m_\mu)^2 \r + \dfrac{\beta }{2 }a(1-a) \sum_{\mu=1}^K \l (M_\mu)^2 \r - \sum_{\mu=1}^K \left( \psi_\mu \l m_\mu \r + \varphi_\mu \l M_\mu \r \right), 
    \end{align}
    which, in the thermodynamic limit,  becomes 
    \begin{align}
        d_t \mathcal{A}_{K, a}(\beta;t) = -\dfrac{\beta}{2} a\sum_\mu \mb_\mu^2 - \dfrac{\beta}{2a^2(1-a)^2} \sum_\mu \bar{M}_\mu^2,
    \end{align}
where we made the sharp choice $\psi^\mu= \beta a\bar m_\mu$ and $\varphi^\mu= \beta a(1-a) \bar M_\mu$.
\end{lemma}

\begin{lemma}
The expression of the interpolating quenched statistical pressure, at $t=0$ and finite-size $N$, reads as 
    \begin{align}
        \mathcal{A}_{N, K, a}(\beta;t=0) =  \mathbb{E} \log \left[ 1+2 \exp \left( \sum_{\mu=1}^K  \beta  \eta^\mu \bar M_\mu \right) \cosh \left(\sum_{\mu=1}^K\beta\xi^\mu \bar m_\mu \right)\right],
    \end{align}
    furthermore, the thermodynamic limit leaves the expression unblemished as  $\lim_{N\to\infty}\mathcal{A}_{N, K, a}(\beta;t=0) =  \mathcal{A}_{K, a}(\beta;t=0)$.
\end{lemma}

Since the proofs of these two Lemmas are rather similar to those of the previous case discussed in Sec. \ref{sec:RS-BEGNN} (actually simpler as there are no overlaps here), we  simply skip them for the sake of conciseness.

\subsubsection{Ghatak-Sherrington model}
As in the previous Section, our goal here is to obtain an explicit expression of the quenched statistical pressure in the low load regime, for finite dilution and under the parallel retrieval assumption.
\newline
Also in this case we apply Guerra interpolation,  therefore --assuming the path to follow is \resub{understood} at this point-- we need to state the following 
\begin{definition}
    Let us consider the interpolating parameter $t \in [0,1]$ and $\psi_\mu, \ \varphi_\mu$, $\mu=1, \hdots, K$ some constants to be set a posteriori to ensure self-averaging of the order parameters. 
    \newline
    The interpolating quenched statistical pressure, at finite volume $N$, reads as 
\begin{align}
    &\mathcal{A}_{N,K,a,S}(\beta;t)=\dfrac{1}{N} \mathbb{E} \log \mathcal{Z}_{N, K, a, S}(\beta;t \vert \bm \xi)  \notag \\
    &=\dfrac{1}{N} \mathbb{E}\log \sum_{\{\bm \sigma\}} \exp \left[\dfrac{\beta N t}{2} \N\sum_{\mu=1}^K (m_\mu)^2 + \dfrac{\beta N t}{2}\Nn \sum_{\mu=1}^K (M_\mu)^2 + (1-t) N\sum_{\mu=1}^K \left( \psi_\mu m_\mu + \varphi_\mu M_\mu\right)\right].
\end{align}
\end{definition}

The underlying {\em modus operandi} is kept the same of the previous section here: also in this case at $t=1$ we recover the original model, instead for $t=0$ we are left with a set of one-body terms whose explicit computation is trivial such that, as a result, we construct a sum rule by the FTC (where the self-averaging of the order parameters makes the evaluation of the derivative of the statistical pressure trivial too).
\newline
Mirroring Assumption \ref{RS-assumption} of the previous section, we impose that the Mattis magnetizations do not fluctuate in the thermodynamic limit or, in other words, we assume that their distributions become delta-peaked on their mean values as $N\to \infty$, therefore easily we can state the main Theorem of this Section, that reads 
\begin{figure}[t!]
    \centering
    \includegraphics[width=15cm]{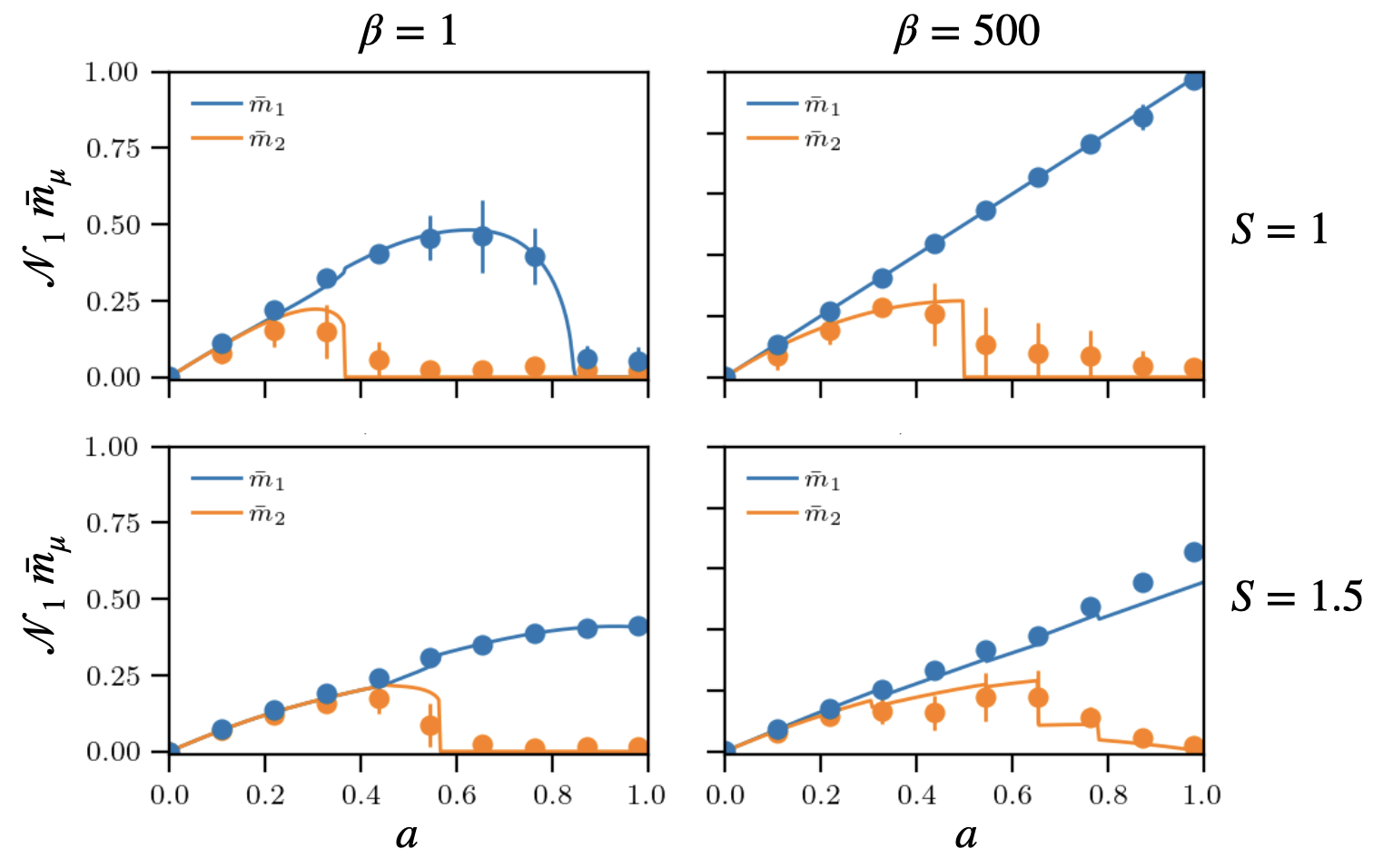}
    \caption{MCMC simulations (dots) versus solutions of the self-consistency equations \eqref{eq:selfMT1} (solid lines), for $K=2$ stored patterns and $N=3000$ spins.  
Top row: $S=1$ (i.e.\ $\boldsymbol\sigma\in\{-1,0,+1\}^N$).  
Bottom row: $S=1.5$ (i.e.\ $\boldsymbol\sigma\in\{-1,-\tfrac{1}{3},+\tfrac{1}{3},+1\}^N$).  
\\
The agreement between simulations and theory is excellent across all regimes.  
The self-consistency equations correctly predict the emergence of a \emph{parallel retrieval} ansatz in certain dilution regions, i.e.\ both Mattis magnetizations remain simultaneously non-zero.  
For small $\beta$ (high temperature), parallel retrieval is observed only at strong dilution (small $a$), regardless of whether $0\in\Omega$.  
For large $\beta$ (low temperature), parallel retrieval persists for all values of $a$ when $0\notin\Omega$, whereas if $0\in\Omega$ it survives only up to a critical dilution level.  
In particular, for $K=2$, $S=1$ and $\beta=500$, the critical threshold is $a_c\simeq 0.5$.
}
    \label{fig:MC_self}
\end{figure}
\begin{theorem}
\label{thm:2}
Under the parallel retrieval ansatz, in the thermodynamic limit $N\to \infty$, in the low-storage of patterns and for finite pattern's dilution values, the expression of the quenched statistical pressure of the Ghatak-Sherrington neural network model, whose Hamiltonian is \eqref{eq:HamiltSspin}, reads in terms of its order and control parameters as 
    \begin{align}
    \mathcal{A}_{K,a,S}(\beta) =& \mathbb{E}_{\bm\xi} \log  \left\{\SOMMA{k=0}{2S} \mathcal{K}(\beta, k, \bm\xi, \bar{\bm M}, \bar{\bm m})\right\}-\dfrac{\beta}{2} \N\SOMMA{\mu=1}{K} \mb_\mu^2 - \dfrac{\beta}{2}\Nn \SOMMA{\mu=1}{K} \bar{M}_\mu^2,
\end{align}
where 
\begin{align}
         s(k) &= -1 + \dfrac{k}{S},
         \\
         \mathcal{K}(\beta, k, \bm\xi, \bar{\bm M}, \bar{\bm m})&=\exp\left\{\beta\SOMMA{\mu=1}{K}  \left[\Big((\xi^\mu)^2-\N\Big)\bar M_\mu s^2(k) + \xi^\mu \bar m_\mu s(k)\right]\right\}
    \end{align}
and the order parameters must fulfill the following self-consistency equations, for $\mu=1, \hdots, K$
\begin{align}
    \mb_\mu=& \mathbb{E}_{\bm\xi} \left[ \dfrac{\SOMMA{k=0}{2S} \mathcal{K}(\beta, k, \bm\xi, \bar{\bm M}, \bar{\bm m}) s(k)}{\SOMMA{k=0}{2S}  \mathcal{K}(\beta, k, \bm\xi, \bar{\bm M}, \bar{\bm m})}\left(\dfrac{\xi^\mu}{\N} \right) \right]\notag \\ 
    \label{eq:selfMT1}
    \\\bar{M}_\mu=& \mathbb{E}_{\bm\xi} \left[ \dfrac{\SOMMA{k=0}{2S}  \mathcal{K}(\beta, k, \bm\xi, \bar{\bm M}, \bar{\bm m})s^2(k)}{\SOMMA{k=0}{2S}  \mathcal{K}(\beta, k, \bm\xi, \bar{\bm M}, \bar{\bm m})}\left(\dfrac{(\xi^\mu)^2-\N}{\Nn}  \right) \right].
    \notag
\end{align}
\end{theorem}

Again the proof of Theorem \ref{thm:2} is a trivial application of the FTC \eqref{eq:FTC} that we achieve, mirroring the previous Section so we do not report it.

In Fig.~\ref{fig:MC_self} we compare MCMC simulations with the self-consistency equations, finding excellent agreement across all regimes.  
The analysis highlights how the presence (or absence) of $0$ in the spin state space $\Omega$ crucially affects the onset of parallel retrieval.  
When $0\in\Omega$ (e.g.\ $S=1$), parallel retrieval is possible only under strong dilution (small $a$), and disappears beyond a critical dilution threshold that depends on temperature (for instance, for $K=2$, $S=1$ and $\beta=500$ we find $a_c\simeq 0.5$).  
By contrast, when $0\notin\Omega$ (e.g.\ $S=1.5$), the system can sustain parallel retrieval across the whole range of $a$ at low temperatures, while at high temperatures the effect remains confined to small $a$.  
This behavior confirms the theoretical predictions of the self-consistency equations and illustrates how dilution and the structure of $\Omega$ jointly control the stability of parallel states; notably, to the best of our knowledge, this is the first work addressing the interplay between graded neuronal response (i.e.\ the structure of $\Omega$) and the network's ability to retrieve multiple patterns in parallel.  

\begin{figure}[t!]
    \centering
    \includegraphics[width=15.5cm]{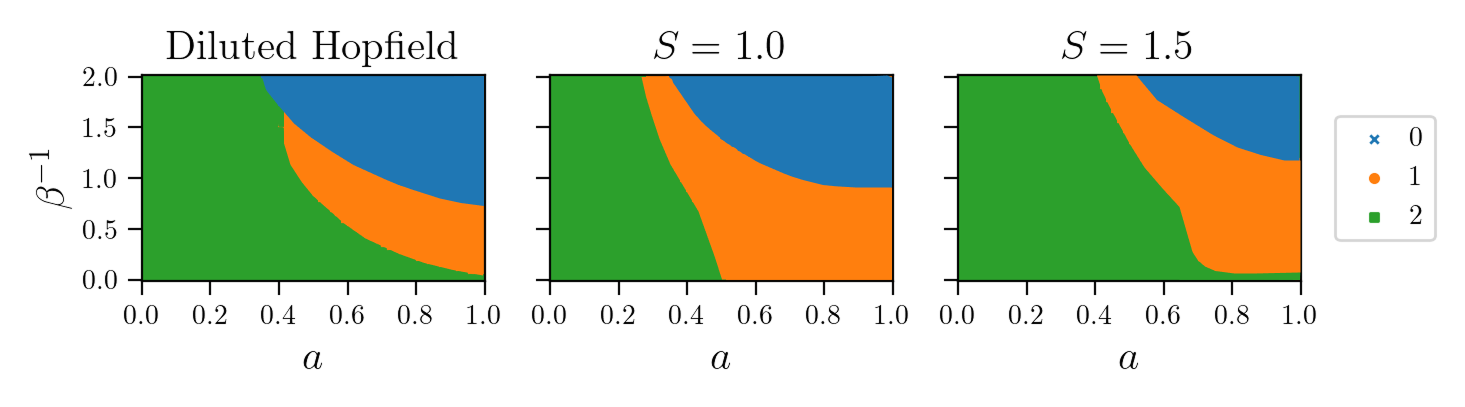}
    \caption{\resub{
Phase diagrams in the $(a,\beta^{-1})$ plane for $K=2$. For each pair $(a,\beta^{-1})$, we numerically solve the self-consistency equations and count the number of overlaps $m_\mu$ such that $|m_\mu|\gtrsim \N(1-a)^{K-1}$, which corresponds to the expected value of the smallest magnetization in the $\beta\to\infty$ limit. Each point is shown with the color and marker reported in the legend (to the right of each row), encoding the number of non-zero self-consistent solutions for that parameter choice. Columns correspond to different models: \textit{(left)} (possibly diluted) Hopfield model with spins $\pm 1$, \textit{(center)}  BEG model with spins $\{-1,0,+1\}$, and \textit{(right)} Ghatak--Sherrington (GS) model with $S=3/2$ and spins $\{-1,-\tfrac{1}{3},+\tfrac{1}{3},+1\}$. At low temperature ($\beta^{-1}\to 0$), the BEG model displays the parallel ansatz (i.e., more than one non-zero solution) only under strong dilution ($a\to 0$), whereas both the Hopfield (left) and GS (right) models exhibit more than one non-zero solution for any $a$, provided the temperature is sufficiently low.
}}
\label{fig:phase_diagram_different_regimes}
\end{figure}

\resub{A complementary, global view is provided by the phase diagrams in Fig.~\ref{fig:phase_diagram_different_regimes}, where we map, in the $(a,\beta^{-1})$ plane, the regions in which one observes one or more non-vanishing self-consistent overlaps, thereby clarifying how the ergodic, serial and hierarchical/parallel regimes are selected by the combined effect of dilution and temperature across the different spin-state spaces.}

 
\subsection{Parallel processing in the medium load regime}
\label{sec:mediumload}
We are now ready to address the case of a sub-linear diverging storage, namely when the number of stored patterns scales with the system size as
\begin{align}
\tilde\alpha = \lim_{N \to \infty} \frac{K}{N^\delta}, \qquad \delta \in (0,1),
\end{align}
and the network is \emph{extremely diluted} but still equipped with a giant component.  
In this regime the dilution parameter $a$ vanishes with $N$ as
\begin{equation}
a \;\to\; a_\gamma(N) = \frac{1}{N^\gamma}, \qquad \gamma\in(0,1).
\end{equation}

Accordingly, the distribution of the pattern entries becomes
\begin{align}
\mathbb{P}(\xi_i^\mu) =
\begin{cases}
\big(1-a_\gamma(N)\big)\,\delta_{\xi_i^\mu,0}
+ \dfrac{a_\gamma(N)}{2S}\SOMMA{\substack{k=0 \\ k\neq S}}{2S}\delta_{\xi_i^\mu,-1+\tfrac{k}{S}} 
& \text{if } S\in \mathbb{N},\\\\
\big(1-a_\gamma(N)\big)\,\delta_{\xi_i^\mu,0}
+ \dfrac{a_\gamma(N)}{2S+1}\SOMMA{k=0}{2S}\delta_{\xi_i^\mu,-1+\tfrac{k}{S}} 
& \text{if } \dfrac{2S+1}{2}\in \mathbb{N}.
\end{cases}
\label{eq:high_dil_prob}
\end{align}
Hence, unlike the finite-dilution setting, here the probability of null entries in the patterns increases with the system size.  
We recall that for $S=1$ this construction reduces to the \resub{BEG} model.

The relevant order parameters in this regime remain the $K$ Mattis magnetizations $m_\mu$ and $M_\mu$, $\mu=1,\dots,K$, with the sole modification that the dilution parameter is now replaced by $a\to a_\gamma(N)$.

In this regime, we can still apply the standard statistical-mechanics machinery, in particular Guerra’s interpolation scheme: since the technical steps of the method are unchanged with respect to the previous sections, we report below only the main results.

\subsubsection{Blume-Emery-Griffiths model}
For the BEG model, Eq. \eqref{eq:high_dil_prob} reduces to
\begin{align}
\mathbb{P}(\xi_i^\mu) =\big(1-a_\gamma(N)\big)\,\delta_{\xi_i^\mu,0}
+ \dfrac{a_\gamma(N)}{2}\delta_{\xi_i^\mu,-1} + \dfrac{a_\gamma(N)}{2}\delta_{\xi_i^\mu,+1}
\end{align}
and the Hamiltonian reads as
\begin{align}
    \mathcal{H}_{N, \delta, \gamma}(\bm \sigma \vert \bm \xi)= - \dfrac{N}{2}a_\gamma(N) \SOMMA{\mu=1}{\tilde\alpha N^{\delta}} m_\mu^2 - \dfrac{N}{2}a_\gamma(N)(1-a_\gamma(N))\SOMMA{\mu=1}{\tilde\alpha N^{\delta}} M_\mu^2.
\end{align}

Following the usual route under the self-averaging of the order parameters, the self-consistency equations that describe the computational behaviour of the BEG network in this setting are
\begin{equation}
    \begin{array}{lll}
         \mb_\mu=& \mathbb{E}_{\bm\xi} \left[ \dfrac{ \tanh \left[\beta G_2(\gamma, \delta, \tilde\alpha,\bm\xi,\bar{\bm m}) \right]}{ 1+\frac{1}{2} e^{-\beta G_1(\gamma, \delta, \tilde\alpha,\bm\xi,\bar{\bm M})} \cosh^{-1}\left[\beta G_2(\gamma, \delta, \tilde\alpha,\bm\xi,\bar{\bm m}) \right]}\left(\dfrac{\xi^\mu}{a_\gamma(N)}\right) \right] ,
    \\\\
    \bar{M}_\mu=& \mathbb{E}_{\bm\xi}  \left[\dfrac{1}{ 1+\frac{1}{2} e^{-\beta G_1(\gamma, \delta, \tilde\alpha,\bm\xi,\bar{\bm M})} \cosh^{-1}\left[\beta G_2(\gamma, \delta, \tilde\alpha,\bm\xi,\bar{\bm m}) \right]}\left(\dfrac{(\xi^\mu)^2-a_\gamma(N)}{a_\gamma(N)(1-a_\gamma(N))}\right)\right],
    \end{array}
    \label{eq:selfMT_high_K}
\end{equation}
where we recall that
\begin{equation}
    G_1(\gamma, \delta, \tilde\alpha,\bm\xi,\bar{\bm M}) = \SOMMA{\mu=1}{\tilde\alpha N^\delta}[(\xi^\mu)^2-a_\gamma(N)] \bar M_\mu \qquad G_2(\gamma, \delta, \tilde\alpha,\bm\xi,\mb) = \SOMMA{\mu=1}{\tilde\alpha N^\delta}\xi^\mu \bar m_\mu.
\end{equation}
Now, performing explicitly the average with respect $\xi^\mu$, with $\mu$ fixed, we get
\footnotesize
\begin{equation}
    \begin{array}{lll}
         \mb_\mu&=& \mathbb{E}_{\bm\xi} \left[ \dfrac{ \tanh \left[\beta\left(\mb_\mu+\SOMMA{\underset{\nu\neq \mu}{\nu=1}}{K}\mb_\nu\xi^\nu\right) \right]}{ 1+\frac{1}{2} \exp \left[- \beta\left(\Mb_\mu(1-a_\gamma(N))+\SOMMA{\underset{\nu\neq \mu}{\nu=1}}{K}\Mb_\nu\Big((\xi^\nu)^2-a_\gamma(N)\Big)\right)\right] \cosh^{-1} \left[\beta\left(\mb_\mu+\SOMMA{\underset{\nu\neq \mu}{\nu=1}}{K}\mb_\nu\xi^\nu\right) \right]} \right] ,
    \\\\
    \bar{M}_\mu&=& \mathbb{E}_{\bm\xi}\left\{ 1+\frac{1}{2} \exp \left[- \beta\left(\Mb_\mu(1-a_\gamma(N))+\SOMMA{\underset{\nu\neq \mu}{\nu=1}}{K}\Mb_\nu\Big((\xi^\nu)^2-a_\gamma(N)\Big)\right)\right] \cosh^{-1} \left[\beta\left(\mb_\mu+\SOMMA{\underset{\nu\neq \mu}{\nu=1}}{K}\mb_\nu\xi^\nu\right) \right]\right\}^{-1}
    \\\\
    &&-\mathbb{E}_{\bm\xi}  \left\{ 1+\frac{1}{2} \exp \left[- \beta\left(-\Mb_\mu a_\gamma(N)+\SOMMA{\underset{\nu\neq \mu}{\nu=1}}{K}\Mb_\nu\Big((\xi^\nu)^2-a_\gamma(N)\Big)\right)\right] \cosh^{-1} \left[\beta\left(\SOMMA{\underset{\nu\neq \mu}{\nu=1}}{K}\mb_\nu\xi^\nu\right) \right]\right\}^{-1}.
    \end{array}
\end{equation}
\normalsize
As $K\sim N^\delta$ with $\delta>0$, we can replace the sum over $\nu\neq \mu$ with the sum over all $\nu$
\footnotesize
\begin{equation}
    \begin{array}{lll}
         \mb_\mu&=& \mathbb{E}_{\bm\xi} \left[ \dfrac{ \tanh \left[\beta\left(\mb_\mu+\SOMMA{\nu=1}{K}\mb_\nu\xi^\nu\right) \right]}{ 1+\frac{1}{2} \exp \left[- \beta\left(\Mb_\mu(1-a_\gamma(N))+\SOMMA{\nu=1}{K}\Mb_\nu\Big((\xi^\nu)^2-a_\gamma(N)\Big)\right)\right] \cosh^{-1} \left[\beta\left(\mb_\mu+\SOMMA{\nu=1}{K}\mb_\nu\xi^\nu\right) \right]} \right] ,
    \\\\
    \bar{M}_\mu&=& \mathbb{E}_{\bm\xi}  \left\{ 1+\frac{1}{2} \exp \left[- \beta\left(\Mb_\mu(1-a_\gamma(N))+\SOMMA{\nu=1}{K}\Mb_\nu\Big((\xi^\nu)^2-a_\gamma(N)\Big)\right)\right] \cosh^{-1} \left[\beta\left(\mb_\mu+\SOMMA{\nu=1}{K}\mb_\nu\xi^\nu\right) \right]\right\}^{-1}
    \\\\
    &&-\mathbb{E}_{\bm\xi}  \left\{ 1+\frac{1}{2} \exp \left[- \beta\left(-\Mb_\mu a_\gamma(N)+\SOMMA{\nu=1}{K}\Mb_\nu\Big((\xi^\nu)^2-a_\gamma(N)\Big)\right)\right] \cosh^{-1} \left[\beta\left(\SOMMA{\nu=1}{K}\mb_\nu\xi^\nu\right) \right]\right\}^{-1}.
    \end{array}
\end{equation}
\normalsize
\begin{figure}[t!]
    \centering
    \includegraphics[width=10cm]{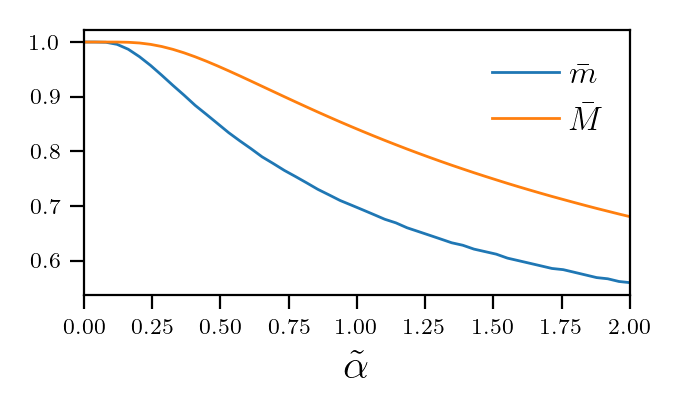}
    \caption{Numerical resolution of the self-consistencies \eqref{eq:high_dil_self} in the  $\beta\to \infty$ limit (plotting, {\em de facto}, eq. \eqref{quattrotrenta}) for both $\bar{m}$ and $\bar{M}$ vs $\tilde\alpha$ as reported in the legend. Note that we are confined to the sub-linear storage --see the discussion before \eqref{spiegaDeltaTilde}-- hence the network has stored  $K =  \alpha N^{\delta}$ (with $0<\delta<1$) but it is retrieving an amount of patterns $\tilde{K} = \tilde{\alpha} N^{\tilde{\delta}}$ with $\tilde{\delta}=\gamma$, that is, the dilution in the pattern determines how many of them can be recalled at once. Note that, apart the extreme-dilution constraint,  these networks share the presence of autapses (i.e. self-synapses, see Figure \ref{fig:NNRBMscheme}), that contribute to motivate such a huge values of $\tilde{\alpha}$ as explained in \cite{Gallius}.}
    \label{fig:high_dil_beta_infinity}
\end{figure}

In this setting, as deepened in \cite{FedeParallelo,Coolen1}, the high dilution forces the pattern to have the same -huge- amount of empty entries and, consequently, the average value of each magnetization is the same for all the retrieved patterns, namely we can write $\mb_\mu=\mb$ for $\mu=1, \cdots, \tilde{K}$ \resub{where $\tilde K=\tilde\alpha N^{\tilde\delta}$ with $0<\tilde\delta \leq\delta$} is the number of no null magnetization. 
This sensibly simplifies the previous equations into
\begin{equation}\label{spiegaDeltaTilde}
    \begin{array}{lll}
         \mb&=& \mathbb{E}_{\bm\xi} \left[ \dfrac{ \tanh \left[\beta\mb\left(1+\SOMMA{\nu=1}{\tilde{K}}\xi^\nu\right) \right]}{ 1+\frac{1}{2} \exp \left[- \beta\Mb\left(1-a_\gamma(N)+\SOMMA{\nu=1}{\tilde{K}}\Big((\xi^\nu)^2-a_\gamma(N)\Big)\right)\right] \cosh^{-1} \left[\beta\mb\left(1+\SOMMA{\nu=1}{\tilde{K}}\xi^\nu\right) \right]} \right] ,
    \\\\
    \bar{M}&=& \mathbb{E}_{\bm\xi}  \left\{  1+\frac{1}{2} \exp \left[- \beta\Mb\left(1-a_\gamma(N)+\SOMMA{\nu=1}{\tilde{K}}\Big((\xi^\nu)^2-a_\gamma(N)\Big)\right)\right] \cosh^{-1} \left[\beta\mb\left(1+\SOMMA{\nu=1}{\tilde{K}}\xi^\nu\right) \right]\right\}^{-1}
    \\\\
    &&-\mathbb{E}_{\bm\xi}  \left\{  1+\frac{1}{2} \exp \left[- \beta\Mb\left(-a_\gamma(N)+\SOMMA{\nu=1}{\tilde{K}}\Big((\xi^\nu)^2-a_\gamma(N)\Big)\right)\right] \cosh^{-1} \left[\beta\mb\left(\SOMMA{\nu=1}{\tilde{K}}\xi^\nu\right) \right]\right\}^{-1}.
    \end{array}
\end{equation}
Now, we can exploit the \resub{CLT} on $\sum_{\nu=1}^{\bar K} \xi^\nu$, we get
\begin{align}
\SOMMA{\nu=1}{\tilde{K}}\xi^\nu&\sim x\sqrt{\tilde\alpha N^{\tilde\delta-\gamma}}\qquad x\sim\mathcal{N}(0,1)
\\
\SOMMA{\nu=1}{\tilde{K}}[(\xi^\nu)^2-a_\gamma(N)]&\sim y\sqrt{\tilde\alpha(1-N^{-\gamma})N^{\tilde\delta-\gamma}}\qquad y\sim\mathcal{N}(0,1),
\end{align}
which represent finite noise contributions if and only if $\tilde\delta=\gamma$
\begin{equation}
    \begin{array}{lll}
         \mb&=& \mathbb{E}_{x, y} \left[ \dfrac{ \tanh \left[\beta\mb(1+x\sqrt{\tilde\alpha}) \right]}{ 1+\frac{1}{2} \exp \left[- \beta\Mb(1+  y\sqrt{\tilde\alpha})\right] \cosh^{-1} \left[\beta\mb(1+x\sqrt{\tilde\alpha}) \right]} \right] ,
    \\\\
    \bar{M}&=& \mathbb{E}_{x,y}  \left\{  1+\frac{1}{2} \exp \left[- \beta\Mb(1+  y\sqrt{\tilde\alpha})\right] \cosh^{-1} \left[\beta\mb(1+x\sqrt{\tilde\alpha})\right]\right\}^{-1}
    \\\\
    &&-\mathbb{E}_{x,y}  \left\{ 1+\frac{1}{2} \exp \left[- \beta\Mb y\sqrt{\tilde\alpha}\right] \cosh^{-1} \left[\beta\mb x\sqrt{\tilde\alpha}\right]\right\}^{-1}.
    \end{array}
    \label{eq:high_dil_self}
\end{equation}

It is possible now to simplify this expression assuming $\beta\to \infty$. Indeed, we get
\begin{equation}\label{quattrotrenta}
    \begin{array}{lll}
         \mb&=& \mathbb{E}_{x, y} \left[ \mathrm{sign} \left[\mb(1+x\sqrt{\tilde\alpha}) \right]\Theta\Big(\mb+\Mb+x\mb\sqrt{\tilde\alpha}+y\Mb\sqrt{\tilde\alpha})\Big) \right] ,
    \\\\
    \bar{M}&=& \text{erf}\left(\dfrac{\mb+\Mb}{ \sqrt{2\tilde\alpha(\mb^2+\Mb^2)}}\right),
    \end{array}
\end{equation}
whose results are reported in Fig. \ref{fig:high_dil_beta_infinity}: as the storage increases the maximal value of the magnetizations gradually smooths down, without any blackout catastrophe: We speculate that the lacking  of the blackout scenario is driven by two contributions, the former being the presence of autapses --as explained in  \cite{Gallius}-- the latter the presence of pattern's dilution, as explained in \cite{Coolen1,Coolen2}.

\subsubsection{Ghatak-Sherrington model}

\begin{figure}[t!]
    \centering
    \includegraphics[width=15cm]{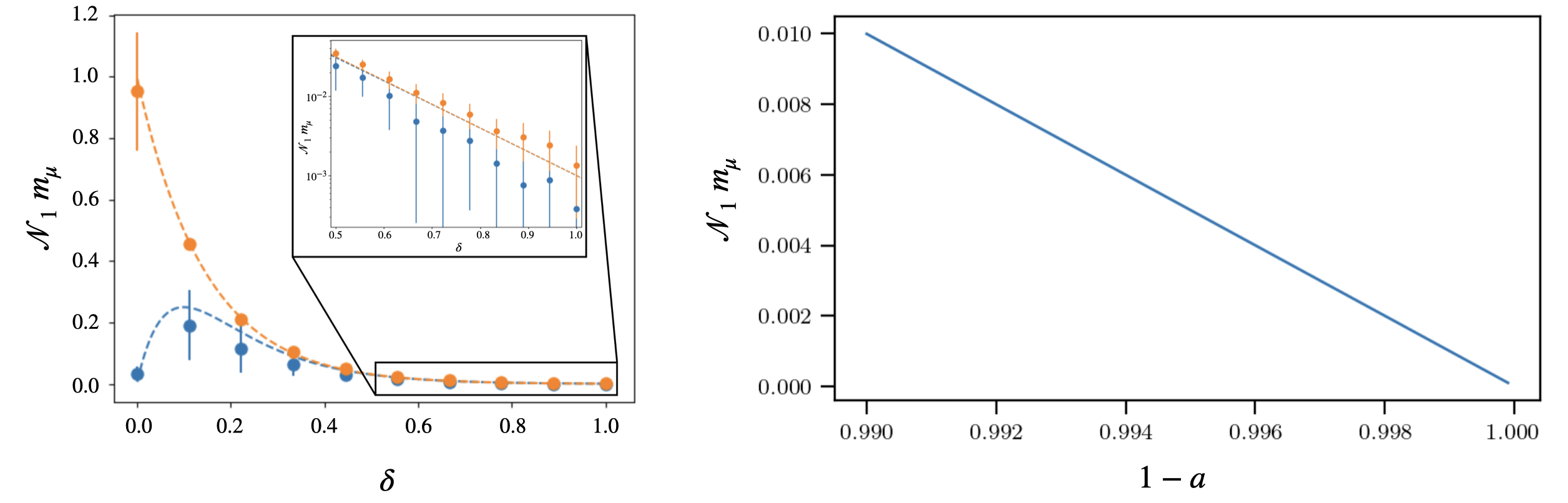}
    \caption{Left panel: first two Mattis magnetizations of the BEG neural network in the extremely diluted case where dilution depends on the network size specifically as $a=1/N^{\delta}$ (so that $\mathbb{P}(\xi_i^\mu=0)=1-1/N^{\delta}$).   As $\delta\to 1$, retrieval is still perfectly preserved: the apparent vanishing of the magnetizations is simply due to the fact that the information content per pattern becomes negligible.   Dotted curves correspond to the theoretical predictions, while dots (with error bars) are MCMC data obtained with $N=1000$, $S=1.0$, $\beta=1000$ and $\tilde\alpha=0.1$.  Left inset: zoom of the same scenario in the extreme high-dilution regime ($\delta>0.5$).  
    \\
    Right panel: overlap between a generic pattern and the neural configuration ($\mathcal{N}_1 m_\mu$) as a function of the average fraction of null pattern entries ($1-a$) in the high-dilution region (small $a$).  Even though the overlap decreases rapidly, this does not signal a loss of retrieval, since the information carried by each pattern is itself reduced.}
    \label{fig:high_dilution_regime}
\end{figure}
 
For the GS we have to follow the exactly same steps done for the BEG model, the only things that change at this level is the value of the normalization function of the $2\times K$ Mattis magnetizations that read as
\begin{align}
    m_\mu(\bm \sigma) = \dfrac{1}{N \N(\gamma, N, S)} \sum_{i=1}^N \xi_i^\mu \sigma_i, \\
    M_\mu (\bm \sigma) = \dfrac{1}{N \Nn(\gamma, N, S)} \sum_{i=1}^N \eta_i^\mu \sigma_i^2,
\end{align}
where, focusing only on the first order in $N^{-\gamma}$, as $N\gg1$ and $\gamma\in(0,1)$, the explicit expressions of $\N(\gamma, N, S)$ and $\Nn(\gamma, N, S)$ are presented in \eqref{eq:norm_GS_high_dil}.

Thus, the Hamiltonian of the models is
\begin{align}
    \mathcal{H}_{N, \alpha, \delta, \gamma, S}(\bm \sigma \vert \bm \xi)= - \dfrac{N}{2} \N(\gamma, N, S)\SOMMA{\mu=1}{\tilde\alpha N^\delta} m_\mu^2 - \dfrac{N}{2}\Nn(\gamma,N,S)\SOMMA{\mu=1}{\tilde\alpha N^\delta} M_\mu^2.
\end{align}

\begin{figure}[t!]
    \centering
    \includegraphics[width=15.5cm]{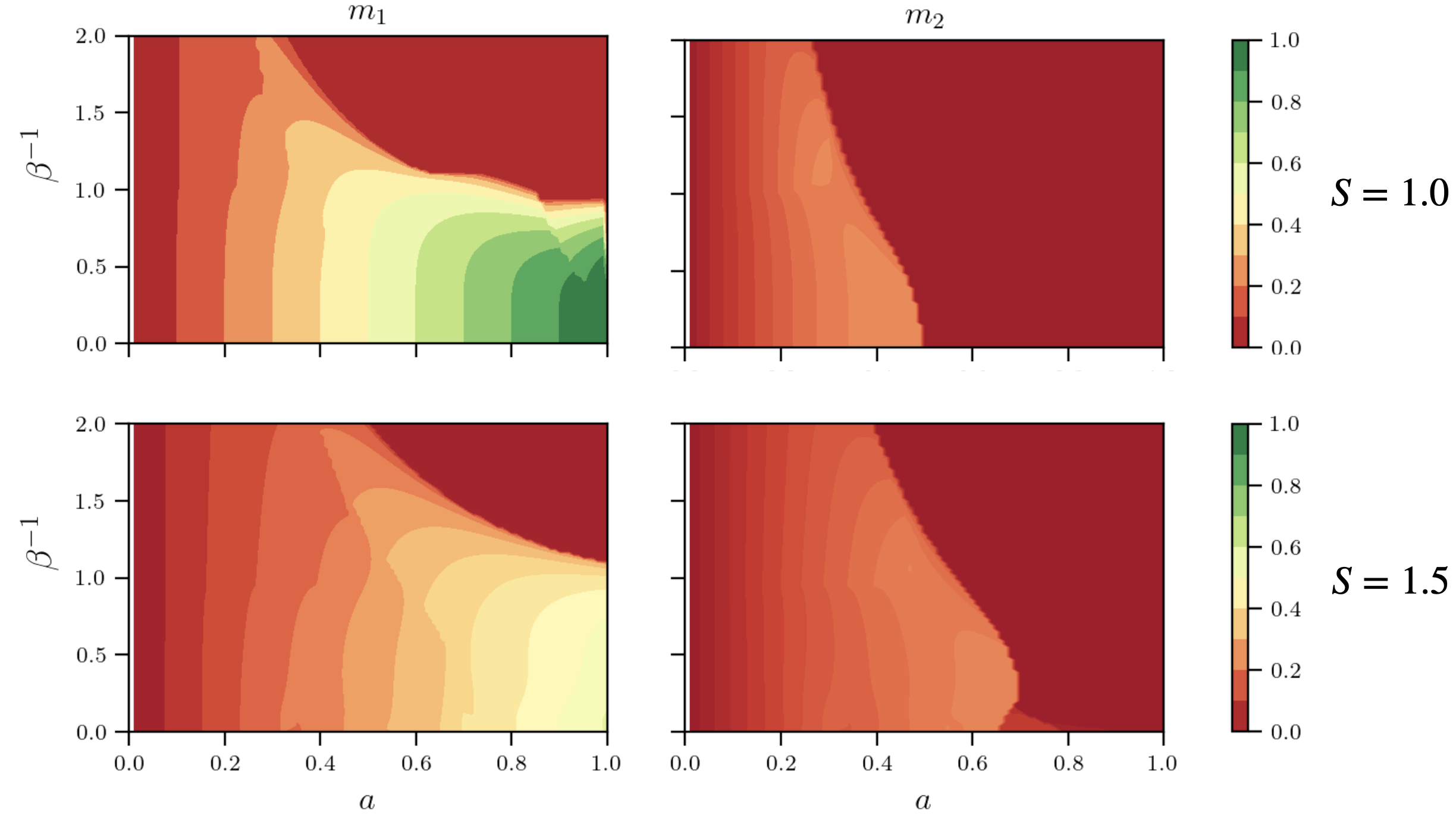}
    \caption{\resub{Numerical self-consistent resolution in the $(a,\beta^{-1})$ plane for $K=2$ of \eqref{eq:selfMT1}.
Top row: $S=1.0$; bottom row: $S=1.5$.
For each parameter pair $(a,\beta^{-1})$ we solve the self-consistency equations and plot the resulting overlap values: the left column shows $m_1$, while the right column shows $m_2$.
The horizontal axis is the dilution parameter $a$ and the vertical axis is the temperature $\beta^{-1}$.
Color intensity encodes the magnitude of the corresponding overlap, as indicated by the color bars on the right.}}
    \label{fig:phase_diagram_ergodicity}
\end{figure}

At this point, the path is \resub{the following}: also in this scenario we want to apply Guerra's interpolation to the quenched statistical pressure in order to express the latter in terms of the order and control parameters of the theory: this allows --by extremization of the statistical pressure with respect to the order parameters-- to obtain their self-consistencies, whose analysis traces their evolution in the space of the control parameters thus highlighting the possible presence of different global information processing regimes. 
Following the usual route under the RS assumption the self consistency equations are
\begin{equation}
    \begin{array}{lll}
         \bar m_\mu=& \mathbb{E}\left\{\dfrac{ \SOMMA{k=0}{2S} \tilde{\mathcal{K}}_{\tilde\alpha,\gamma,\delta,N}(\beta, k, \bm\xi, \bar{\bm M}, \bar{\bm m}) s(k) }{ \SOMMA{k=0}{2S}\tilde{\mathcal{K}}_{\tilde\alpha,\gamma,\delta,N}(\beta, k, \bm\xi, \bar{\bm M}, \bar{\bm m})}\left(\dfrac{\xi^\mu}{\N(\gamma,N,S)}\right) \right\}
        \\\\
        \bar M_\mu=& \mathbb{E}\left\{\dfrac{ \SOMMA{k=0}{2S} \tilde{\mathcal{K}}_{\tilde\alpha,\gamma,\delta,N}(\beta, k, \bm\xi, \bar{\bm M}, \bar{\bm m}) s^2(k) }{ \SOMMA{k=0}{2S}\tilde{\mathcal{K}}_{\tilde\alpha,\gamma,\delta,N}(\beta, k, \bm\xi, \bar{\bm M}, \bar{\bm m})}\left(\dfrac{(\xi^\mu)^2-\N(\gamma,N,S)}{\Nn(\gamma,N,S)}\right) \right\}. 
    \end{array}
\end{equation}
where
\begin{equation}
    \begin{array}{lll}
         \tilde{\mathcal{K}}_{\tilde\alpha,\gamma,\delta,N}(\beta, k, \bm\xi, \bar{\bm M}, \bar{\bm m})&=&\exp\left\{\beta\SOMMA{\mu=1}{\tilde\alpha N^{\delta}}  \left[\Big((\xi^\mu)^2-\N(\gamma, N, S)\Big)\bar M_\mu s^2(k) + \xi^\mu \bar m_\mu s(k)\right]\right\}.
    \end{array}
\end{equation}

Outcomes from these self-consistency relations are presented in Fig. \ref{fig:high_dilution_regime}: note that, while it seems that the retrieval goes to zero as $\delta \to 1$, it is actually the information content per pattern that is becoming negligible (and this is the core reason that drives the network to switch from serial to parallel processing),  while the magnetizations are actually perfectly retrieving and an extensive number of these is raised --with the same amplitude-- up to saturation of the  alignment of the neural configuration. \resub{What we mean is that dilution is essential for allowing genuine parallel processing. To see this, consider a simple example with two patterns, each carrying information only on a distinct half of the system and being blank elsewhere. Unlike in the standard Hopfield model, where retrieving one pattern occupies all available degrees of freedom and precludes the retrieval of any other, here neither pattern can be fully retrieved. This leaves unused sites that the network can allocate to the other pattern. Consequently, the optimal configuration is one in which both patterns are retrieved as far as their non-blank entries allow, so that each reaches its maximal attainable magnetization rather than suppressing the other\footnote{For more details one can see \cite{FedeParallelo}, where multitasking capabilities in diluted associative neural networks are discussed extensively.}.}

As a final remark, we comment on ergodicity breaking (please  refer to Fig.~\ref{fig:phase_diagram_ergodicity}). Starting from the high fast-noise regime (small $\beta$), we monitor when the Mattis magnetizations detach from zero; the corresponding critical value of $\beta$ marks the onset of ergodicity breaking. Below this threshold the dynamics can effectively select and remain trapped in the relevant valleys of the energy landscape, whose minima correspond to retrieval states. Importantly, owing to the presence of the quartic interaction, the ergodic-to-nonergodic transition is generically discontinuous (first order in the Ehrenfest classification), with notable exceptions only in special limits, most prominently $a\to 1$, where dilution becomes irrelevant.

\section{Conclusions and outlooks}
In this work we have provided a rigorous statistical-mechanical analysis of neural networks whose Cost function is inspired by {\em inverse freezing} spin glass models, namely the Blume-Emery-Griffiths (BEG) and the Ghatak-Sherrington (GS) neural networks, extending and consolidating previous heuristic results obtained within the replica framework. 
\newline
By employing Guerra’s interpolation technique, we established exact expressions for the quenched free energy in different storage regimes, thereby placing the retrieval properties of these networks on a firm mathematical footing. Our investigation highlights three central outcomes. 
\newline
The first is that, in the high storage regime, all the previous results on serial retrieval (that is the standard pattern recognition phenomenon) have been punctually confirmed by the Guerra's approach under the request that the order parameters self-average around their unique means (that is, at the replica symmetric level of description). 
\newline
Then, we showed how the introduction of sparsity and dilution in the patterns may drastically modify the retrieval phenomenology. In particular, we proved that these networks spontaneously switch between serial and parallel recall, with multitasking capabilities becoming stronger as dilution increases: this feature not only enriches the theoretical phase diagram of associative networks but also resonates with biological evidence of sparse coding and distributed processing. 
\newline
Finally, by enlarging the neuronal state space to multi-level variables, so moving from the BEG to the GS model,  we demonstrated that these neural networks provide a natural setting for modeling graded responses also coping with parallel retrieval, thereby bridging simplified spin-glass models with more realistic descriptions of neural computation.

\medskip
From a conceptual standpoint, the present analysis supports the view that sparsity and multi-state neurons are not merely technical generalizations, but structural ingredients that enhance the expressive power and computational versatility of associative memories.
\newline
Future developments may proceed along several directions. \resub{It would be natural to investigate Replica Symmetry Breaking corrections and to inspect their consequences for the retrieval properties, with the aim of determining whether the onset of symmetry breaking affects the stability of the retrieval states, modifies the boundaries of the retrieval region, or gives rise to new metastable configurations, and to establish whether such an onset is signalled by the emergence of an analogue of the de Almeida-Thouless line.}
\resub{Furthermore, we should extend the present results to finite connectivity architectures, so as to enable parallel retrieval in the high storage setting. In addition, it would be worthwhile to explore the effects of progressively densification of the network, in order to assess how increasing connectivity influences the emergence, stability, and interplay of serial and parallel retrieval regimes.}
Last but not least,  we should relax the assumption of already defined patterns to store and let the network learn them by presenting solely corrupted examples: we plan to report soon on these extensions.

\newpage
\appendix
\section{Computations in Lemma \ref{lemma1}}
\label{app:conti}
From the application of the definition by $t$ of the interpolating quenched statistical pressure we get: 
\begin{align}
    &d_t \mathcal{A}_{N, K, a}(\beta;t)
    =  \mathbb{E}  \left\{ \dfrac{\beta  a}{2}\omega_t( m_1^2) + \dfrac{\beta  a(1-a)}{2}  \omega(M_1^2) \notag \right. \notag \\
    &+ \dfrac{1}{2 N}\sqrt{\dfrac{\beta }{tN}} \sum_{i=1}^N \sum_{\mu >1}^K \tilde \xi_i^\mu \omega_t(\sigma_i z_\mu) \notag \\
    & +  \dfrac{1}{2N}\sqrt{\dfrac{\beta }{N t}} \sum_{i=1}^N \sum_{\mu >1}^K \tilde \eta_i^\mu \omega_t(\sigma_i^2 \tilde z_\mu) - \left(\psi  \omega_t(m_1) + \varphi  \omega_t(M_1) \right)\notag \\
    &- \dfrac{1}{2N} \left[ F \sum_{\mu=2}^K \omega_t(z_\mu^2) + G \sum_{\mu=2}^K \omega_t(\tilde z_\mu^2) + H_1 \sum_{i=1}^N \omega_t(\sigma_i^2)+ H_2 \sum_{i=1}^N \omega_t(\sigma_i^2) \right] \notag \\
    &\left.+ \dfrac{1}{2N\sqrt{1-t}} \left[ A \sum_{i=1}^N K_i \omega_t(\sigma_i) + B \sum_{\mu=2}^K \bar{J}_\mu \omega_t(z_\mu) + C \sum_{\mu=2}^K \tilde J_\mu \omega_t(\tilde z_\mu) + E \sum_i \tilde K_i \omega_t(\sigma_i^2)\right]\right\}. \notag 
\end{align}
Now applying both the definition of the quenched general average ( Def.~\ref{def:global}) and the definition of the order parameters ( Def.~\ref{def:ordparam}), we find
\begin{align}
    &d_t \mathcal{A}_{N, K, a}(\beta;t) = \dfrac{\beta a}{2} \l (m_1)^2  -\dfrac{2\psi}{\beta a} m_1\r + \dfrac{\beta a(1-a)}{2} \l (M_1)^2 -\dfrac{2\varphi}{\beta a(1-a)}M_1\r \notag
    \\
    &- \dfrac{K}{2N} \left[ F  \l p_{11}\r + G \sum_{\mu=2}^K \l\tilde p_{11}\r + (H_1+H_2) \sum_{i=1}^N \l q_{11}\r \right]\notag 
    \\
    &+ \dfrac{1}{2 N}\sqrt{\dfrac{\beta }{tN}} \sum_{i=1}^N \sum_{\mu >1}^K \mathbb{E}[\tilde \xi_i^\mu \omega_t(\sigma_i z_\mu)]  +  \dfrac{1}{2N}\sqrt{\dfrac{\beta }{N t}} \sum_{i=1}^N \sum_{\mu >1}^K \mathbb{E}[\tilde \eta_i^\mu \omega_t(\sigma_i^2 \tilde z_\mu)] \notag 
    \\
    &\left.+ \dfrac{1}{2N\sqrt{1-t}} \left[ A \sum_{i=1}^N \mathbb{E}[K_i \omega_t(\sigma_i)] + B \sum_{\mu=2}^K \mathbb{E}[\bar{J}_\mu \omega_t(z_\mu)] + C \sum_{\mu=2}^K \mathbb{E}[\tilde J_\mu \omega_t(\tilde z_\mu)] + E \sum_i \mathbb{E}[\tilde K_i \omega_t(\sigma_i^2)]\right]\right\} \notag ,
\end{align}
For the remaining terms, we need to apply the a simple Gaussian integration by parts formula
\begin{equation}
    \mathbb{E}_{\bm J}( J \omega_{\bm J}(f(\bm J)) ) = \mathbb{E}_{\bm J} \Bigg( \omega_{\bm J}\left(\frac{\partial}{\partial \bm J} f(\bm J)\right) \Bigg),
\label{eq:wick_theorem}
\end{equation}
valid for any smooth function $f(J)$ of a centered unit Gaussian variable $J$. As the accounts to be made are similar, we only show those for the first term:
\begin{align}
    &\dfrac{1}{2 N}\sqrt{\dfrac{\beta }{tN}} \sum_{i=1}^N \sum_{\mu >1}^K \mathbb{E}[\tilde \xi_i^\mu \omega_t(\sigma_i z_\mu)] =\dfrac{1}{2 N}\sqrt{\dfrac{\beta }{tN}} \sum_{i=1}^N \sum_{\mu >1}^K \mathbb{E}[ \partial_{\tilde\xi_i^\mu}\omega_t(\sigma_i z_\mu)]
    \\
    &=\dfrac{\beta }{2N^2} \sum_{i=1}^N \sum_{\mu >1}^K \Bigg(\mathbb{E}[ \omega_t(\sigma_i^2 z_\mu^2)]- \mathbb{E}[ \omega_t^2(\sigma_i z_\mu)]\Bigg)= \dfrac{\beta K}{2 N}\Big(\l q_{11}p_{11}\r-\l q_{12}p_{12}\r\Big),
    \label{eq:app2}
\end{align}
for the definition of the order parameters. By Applying the same steps to all the remaining terms, one get
\begin{align}
    &d_t \mathcal{A}_{N, K, a}(\beta;t) = \dfrac{\beta a}{2} \l (m_1)^2  -\dfrac{2\psi}{\beta a} m_1\r + \dfrac{\beta a(1-a)}{2} \l (M_1)^2 -\dfrac{2\varphi}{\beta a (1-a)}M_1\r \notag
    \\
    &- \dfrac{K}{2N} \left[ F  \l p_{11}\r + G \sum_{\mu=2}^K \l\tilde p_{11}\r + (H_1+H_2) \sum_{i=1}^N \l q_{11}\r \right]\notag 
    \\
    &+ \dfrac{\beta K}{2 N}\Big(\l q_{11}p_{11}\r-\l q_{12}p_{12}\r\Big)  +  \dfrac{\beta K}{2 N}\Big(\l q_{11}\tilde{p}_{11}\r-\l \tilde{q}_{12}\tilde{p}_{12}\r\Big) \notag 
    \\
    &\left.+ \dfrac{A^2}{2} \Big(\l q_{11}\r-\l q_{12}\r\Big) + \dfrac{B^2 K}{2N}  \Big(\l p_{11}\r-\l p_{12}\r\Big) + \dfrac{C^2 K}{2N} \Big(\l \tilde p_{11}\r-\l \tilde p_{12}\r\Big) + \dfrac{E^2}{2}  \Big(\l q_{11}\r-\l \tilde q_{12}\r\Big)\right\} \notag ,
\end{align}

\section{Proof of Lemma \ref{lemma2}}
\label{app:ob}

\begin{align}
    \mathcal{A}_{N,K,a}(\beta;t=0) =& \dfrac{1}{N} \mathbb{E} \log \left\{ \int \mathcal{D}(z\tilde z) \exp \left[ B \sum_{\mu=2}^K \bar{J}_\mu z_\mu + \dfrac{F}{2} \sum_{\mu=2}^K z_\mu^2 \right] \exp \left[ C \sum_{\mu=2}^K \tilde J_\mu \tilde z_\mu + \dfrac{G}{2} \sum_{\mu=2}^K \tilde z_\mu^2 \right] \right. \notag \\
    &\left.\times \sum_{\{\bm \sigma\}} \exp \left[ \psi N m_1 + \varphi N M_1 + A \sum_{i=1}^N J_i \sigma_i + E \sum_{i=1}^N \tilde J_i \sigma_i^2 + \dfrac{(H_1+H_2)}{2}\sum_{i=1}^N \sigma_i^2 \right]\right\}.
\end{align}
Let us compute the terms inside the logarithm. 
\\
The first one dependently only on $\bm\sigma$ is
\begin{align}
    &\sum_{\bm \sigma} \exp \left[ \psi N  m_1 + \varphi N  M_1 + A \sum_{i=1}^N J_i \sigma_i + E \sum_{i=1}^N \tilde J_i \sigma_i^2 + \dfrac{H_1+H_2}{2}\sum_{i=1}^N \sigma_i^2 \right] \notag \\
    &= \sum_{\bm \sigma} \prod_{i=1}^N\exp \left[  \left(\dfrac{\psi}{a} \xi_i^1+ A  J_i \right) \sigma_i + \left(E \tilde J_i + \dfrac{H}{2} +  \dfrac{\varphi}{a(1-a)} \eta_i^1\right) \sigma_i^2 \right] \notag \\
    &=\prod_{i=1}^N\sum_{\sigma_i \in \{ 0, \pm 1\}} \exp \left[  \left(\dfrac{\psi}{a} \xi_i^1+ A  J_i \right) \sigma_i + \left(E \tilde J_i + \dfrac{H_1+H_2}{2} +  \dfrac{\varphi}{a(1-a)} \eta_i^1\right) \sigma_i^2 \right] \notag \\
    &= \prod_{i=1}^N \left\{1+ 2\exp \left( E \tilde J_i + \dfrac{H_1+H_2}{2} +  \dfrac{\varphi}{a(1-a)} \eta_i^1\right)\cosh \left( \dfrac{\psi}{a} \xi_i^1+ A  J_i \right) \right\}.
\end{align}
The second term dependently only on $\bm z$ and $\tilde{\bm z}$ is
\begin{align}
    \int \mathcal D (z \tilde z) \exp \left[ B \sum_{\mu=2}^K \bar{J}_\mu z_\mu + \dfrac{F}{2} \sum_{\mu=2}^K z_\mu^2 +C \sum_{\mu=2}^K \tilde J_\mu \tilde z_\mu + \dfrac{G}{2} \sum_{\mu=2}^K \tilde z_\mu^2\right] =
    \\
    = \prod_{\mu >1}^K \left[(1-F)(1-G)\right]^{-1/2}\exp \left( \dfrac{B^2 \bar J_\mu^2}{2(1-F)} \right)\exp \left( \dfrac{C^2 \tilde J_\mu^2}{2(1-G)}\right).
\end{align}
For the properties of the logarithm,  we can write that 
\begin{align}
    \mathcal{A}_N(t=0 \vert \beta) =& \mathbb{E} \left\{\log \left[1+ 2\exp \left( E \tilde J + \dfrac{H_1+H_2}{2} +  \dfrac{\varphi}{a(1-a)} \eta^1\right)\cosh \left( \dfrac{\psi}{a} \xi^1+ A  J\right) \right]\right\}\notag \\
    &+ \dfrac{1}{N} \sum_{\mu>1}^K \mathbb{E}\left[ \log \exp \left( \dfrac{B^2 \bar J_\mu^2}{2(1-F)} \right)\right] + \dfrac{1}{N} \sum_{\mu>1}^K \mathbb{E}\left[ \log \exp \left( \dfrac{C^2 \tilde J_\mu^2}{2(1-G)}\right) \right]\notag \\
    &-\dfrac{1}{2N}\SOMMA{\mu>1}{K}\log(1-F)-\dfrac{1}{2N}\SOMMA{\mu>1}{K}\log(1-G)\notag
    \\
    =& \mathbb{E} \left\{\log \left[1+ 2\exp \left( E \tilde J + \dfrac{H_1+H_2}{2} +  \dfrac{\varphi}{a(1-a)} \eta^1\right)\cosh \left( \dfrac{\psi}{a} \xi^1+ A  J \right) \right]\right\}\notag \\
    & -\dfrac{K}{2N}\log(1-G) -\dfrac{K}{2N}\log(1-F)+\dfrac{B^2 K }{2N(1-F)} + \dfrac{C^2 K}{2N(1-G)}.
    \label{eq:OBapp}
\end{align}
In the thermodynamic limit, using the relation $\alpha=\lim_{N\to\infty}K/N$, we get
\begin{align}
    \mathcal{A}^{RS}_{\alpha,a}(\beta;t=0) =& \mathbb{E} \left\{\log \left[1+ 2\exp \left( E \tilde J + \dfrac{H_1+H_2}{2} +  \dfrac{\varphi}{a(1-a)} \eta^1\right)\cosh \left( \dfrac{\psi}{a} \xi^1+ A  J \right) \right]\right\} \notag \\
    &-\dfrac{\alpha}{2}\log(1-G) -\dfrac{\alpha}{2}\log(1-F)+\dfrac{\alpha B^2 }{2(1-F)} + \dfrac{\alpha C^2 }{2(1-G)}
    \label{eq:OBapp_1}
\end{align}
which is well defined in the thermodynamic limit.

\section{General Momenta Calculus}
\label{app:momenta}
We start from the distribution
\begin{align}
\mathbb{P}(\xi_i^\mu) =
    \begin{cases}
        \big(1-a\big)\,\delta_{\xi_i^\mu,0}+ \dfrac{a}{2S}\SOMMA{\substack{k=0 \\ k\neq S}}{2S}\delta_{\xi_i^\mu,-1+\tfrac{k}{S}} & \text{if } S\in \mathbb{N},
        \\\\
        \big(1-a\big)\,\delta_{\xi_i^\mu,0}+ \dfrac{a}{2S+1}\SOMMA{k=0}{2S}\delta_{\xi_i^\mu,-1+\tfrac{k}{S}} & \text{if } \dfrac{2S+1}{2}\in \mathbb{N}.
    \end{cases}     
\end{align}
From this definition it is straightforward to compute explicitly the normalization factors introduced in \eqref{eq:GS_normalization}, which read as
\begin{equation}
\begin{array}{lll}
     \N(a,S)=\mathrm{Var}[\xi_i^\mu]=\begin{cases}
        \dfrac{a (1 + S) (1 + 2 S)}{6 S^2},&\text{if }\;S\in\mathbb{N},
        \\\\
        \dfrac{a (1 + S)}{3 S},&\text{if }\;\dfrac{2S+1}{2}\in\mathbb{N},
    \end{cases}
\\\\
    \Nn(a,S)=\mathrm{Var}[\eta_i^\mu]=\begin{cases}
        \N(a,S)\,\dfrac{18 S (S+1)+6-5 a (2 S+1) (S+1)}{30 S^2},&\text{if }\;S\in\mathbb{N},
        \\\\
        \N(a,S)\,\dfrac{ S (S+1)(9-5 a)-3}{15 S^2},&\text{if }\;\dfrac{2S+1}{2}\in\mathbb{N}.
    \end{cases} 
\end{array}
\label{eq:norm_GS_low_dil}
\end{equation}
The extension to the high-dilution regime is straightforward: it only requires the replacement 
\[
a \;\to\; a_{\gamma}(N)=N^{-\gamma}.
\] 
Keeping only the first-order term in $N^{-\gamma}$ for $N\gg 1$, equation \eqref{eq:norm_GS_low_dil} reduces to
\begin{equation}
\begin{array}{lll}
     \N(\gamma, N,S)= N^{-\gamma}\dfrac{(1+S)}{3S}\begin{cases}
        \dfrac{ 1 + 2 S}{2 S},&\text{if }\;S\in\mathbb{N},
        \\\\
        1,&\text{if }\;\dfrac{2S+1}{2}\in\mathbb{N},
    \end{cases}
\\\\
    \Nn(\gamma, N,S)= N^{-\gamma}\dfrac{ (1 + S) \big(3 S (1 + S)-1\big)}{15 S^3}\begin{cases}
        \dfrac{  1+2 S }{2 S},&\text{if }\;S\in\mathbb{N},
        \\\\
        1,&\text{if }\;\dfrac{2S+1}{2}\in\mathbb{N}.
    \end{cases}
\end{array}
\label{eq:norm_GS_high_dil}
\end{equation}

\section{Numerical simulations and pseudo-code}
Starting from \eqref{eq:HamiltSspin} we can rewrite it as
\begin{align}
\mathcal H^{(GS)}_{N,K,a,S}(\boldsymbol\sigma \mid \boldsymbol\xi)
&=-\frac{1}{2}\sum_{i,j=1}^N J_{ij}\,\sigma_i\sigma_j
-\frac{1}{2}\sum_{i,j=1}^N K_{ij}\,\sigma_i^2\sigma_j^2,
\end{align}
where we set
\begin{equation}
    J_{ij}= \frac{1}{N\,\N(a,S)}\sum_{\mu=1}^K \xi_i^\mu\xi_j^\mu,\qquad K_{ij}=\frac{1}{N\,\Nn(a,S)}\sum_{\mu=1}^K \eta_i^\mu\eta_j^\mu.
\end{equation}
Now exploiting the mean-field structure of the network we have
\begin{align}
\mathcal H^{(GS)}_{N,K,a,S}(\boldsymbol\sigma \mid \boldsymbol\xi)
&=-\SOMMA{i=1}{N}h_1^{(1)}(\boldsymbol\sigma \mid \boldsymbol\xi,\boldsymbol\eta)\sigma_i
-\SOMMA{i=1}{N}h_i^{(2)}\sigma_i^2
\end{align}
where
\begin{align}
h^{(1)}_{i}(\boldsymbol\sigma \mid \boldsymbol J)
=\frac{1}{2}\sum_{j=1}^N J_{ij}\,\sigma_j,\qquad h^{(2)}_{i}(\boldsymbol\sigma \mid K)
=\frac{1}{2}\sum_{j=1}^N  K_{ij}\,\sigma_j^2.
\end{align}
The neurons are updated asynchronously according to the transition probability
\begin{equation}
    \mathbb{P}(\sigma^{'}_i=s(k)|\bm\sigma, \bm\xi) = \dfrac{\exp\left[\beta h^{(1)}_{i}(\boldsymbol\sigma \mid \boldsymbol J)s(k)+\beta h^{(2)}_{i}(\boldsymbol\sigma \mid \boldsymbol K)s^2(k)\right]}{\SOMMA{k=0}{2S}\exp\left[\beta h^{(1)}_{i}(\boldsymbol\sigma \mid \boldsymbol J)s(k)+\beta h^{(2)}_{i}(\boldsymbol\sigma \mid \boldsymbol K)s^2(k)\right]}
    \label{eq:updating_rules_MCMC}
\end{equation}
where $s(k)=-1+k/S$.

\medskip
\begin{algorithm}[H]
\caption{MCMC: serial updating  \label{algo:MCMC}}
\KwIn{Patterns $\{\xi_{i}^{\mu}\}_{i=1, \cdots, N}^{\mu=1, \cdots, K}$, starting configurations $\bm\sigma^{(1)}(t=0)$, number of serial dynamic steps $ N_s $, thermal noise $T$}
\KwOut{Final neuronal configuration $\bm\sigma^{(1)}(t=N_s)$}

Set $n = 0$\;

\Repeat{$n = N_s$}{
     Choose randomly $i$ from $[1,\cdots, N]$;\\
     For each $k=0, \cdots, 2S$, compute  the probability of occurrence  $$p^{(k)}_i = \mathbb{P}\left(\sigma_i(n+1)=s(k)|\bm\sigma(t=(n)), \bm\xi\right);$$ 
     Extract randomly using the vector of probability $p^{(k)}_i$ one of the $2S+1$ possible value of $s(k)$;\\
     Update the $i$-th component of the spin configuration $\sigma_i(t=n+1) = s^*(k) $;\\
    $ n = n + 1 $\;
}
\end{algorithm}

\newpage
\section{List of symbols}

\begin{table}[h!]
\centering
\renewcommand{\arraystretch}{1.25}
\begin{tabular}{ll}
\hline
\textbf{Symbol} & \textbf{Definition} \\
\hline

$\sigma_i$ & State of neuron $i$ (BEG: $\{-1,0,+1\}$; GS: graded values in $\{-1 + k/S\}$). \\

$\xi_i^\mu$ & Entry $i$ of pattern $\mu$ (stored input). \\
$\eta_i^\mu$ & Activity field associated with pattern $\mu$: $(\xi_i^\mu)^2 - \mathcal N_1$ (which is $(\xi_i^\mu)^2 - a$ in BEG model). \\

$a$ & Pattern dilution parameter. \\
$S$ & Resolution parameter of the GS model controlling the graded neural response. \\

$K$ & Number of stored patterns. \\
$N$ & Number of neurons in the network. \\

$\beta$ & Inverse temperature (thermal noise level), $\beta = 1/T$. \\
$\alpha$ & Memory load: $\alpha = \lim_{N\to\infty} K / N^\delta$. \\
$\delta$ & Exponent specifying the storage regime ($0<\delta\le1$). \\
$\mathcal N_1(a,S)$ & Variance of $\xi_i^\mu$ (in BEG is equal to $a$). \\ 
$\mathcal N_2(a,S)$ & Variance of $\eta_i^\mu$ (in BEG is equal to $a(1-a)$. \\

$\Omega$ & Configuration space of the neurons. \\

$\mathcal H(\sigma \mid \xi)$ & Hamiltonian of the model. \\
$\mathcal Z(\beta \mid \xi)$ & Partition function. \\
$\omega(\cdot)$ & Gibbs expectation (thermal average). \\
$\langle \cdot \rangle$ & Global average over both thermal noise and quenched disorder. \\

$m_\mu$ & Mattis magnetisation aligned with pattern $\mu$. \\
$M_\mu$ & Mattis activity magnetisation aligned with $\eta^\mu$. \\

$q_{\ell\ell'}$ & Replica overlap between replicas $\ell$ and $\ell'$. \\
$\tilde q_{\ell\ell'}$ & Overlap of squared neural activities $\sigma^2$ between replicas. \\

$p_{\ell\ell'}$ & Overlap between $z_\mu$ variables. \\
$\tilde p_{\ell\ell'}$ & Overlap between $\tilde z_\mu$ variables. \\
$\tilde{\bar{p}}$ & Replica-symmetric average of $\tilde p_{\ell\ell'}$. \\
$\bar p$ & Replica-symmetric average of $p_{\ell\ell'}$. \\
$\bar P$ & Replica-symmetric average of $p_{\ell\ell}$. \\
$\tilde {\bar P}$ & Replica-symmetric average of $\tilde p_{\ell\ell}$. \\
$Q$ & Replica-symmetric average of $q_{\ell\ell}$ . \\
$\tilde Q$ & Replica-symmetric average of $\tilde q_{\ell\ell}$ (in BEG is equal to $Q$). \\
$\bar q$ & Replica-symmetric average of $q_{\ell\ell'}$. \\
$\tilde{\bar{q}}$ & Replica-symmetric average of $\tilde q_{\ell\ell'}$. \\

$t$ & Guerra interpolation parameter. \\

$s(k)$ & Allowed GS neuron state: $s(k) = -1 + k/S$. \\

$\mathcal A(\beta,\alpha)$ & Quenched statistical pressure. \\
$f(\beta,\alpha)$ & Intensive free energy, related via $\mathcal A = -\beta f$. \\

\hline
\end{tabular}
\end{table}

\newpage
\acknowledgments

LA and AB acknowledge the PRIN 2022 grant {\em Statistical Mechanics of Learning Machines} number 20229T9EAT funded by the Italian Ministry of University and Research (MUR) in the framework of European Union - Next Generation EU.\\
LA acknowledges funding also by the PRIN 2022 grant {\em “Stochastic Modeling of Compound Events (SLIDE)”} n. P2022KZJTZ funded by the Italian Ministry of University and Research (MUR) in the framework of European Union - Next Generation EU and also by “Patto Territoriale del Sistema Universitario Pugliese - Open Apulian University” (CUP F61B230003700006).
\newline
AB acknowledges further support by Sapienza Università di Roma (via the grant {\em Statistical learning theory for generalized Hopfield models}), prot. n. RM12419112BF7119.
\newline
ENMC thanks the PRIN 2022 project ``Mathematical Modelling of Heterogeneous Systems (MMHS)", financed by the European Union - Next Generation EU,
CUP B53D23009360006, Project Code 2022MKB7MM, PNRR M4.C2.1.1.
\newline
All the authors are members of the GNFM group within INdAM which is acknowledged too. 


\begin{thebibliography}{10}

\bibitem{AABO-JPA2020}
E.~Agliari, L.~Albanese, A.~Barra, and G.~Ottaviani.
\newblock Replica symmetry breaking in neural networks: a few steps toward rigorous results. 
\newblock {\em J. Phys. A: Math. $\&$ Theor.} 53, 415005, 2020.

\bibitem{BarraLenka}
E.~Agliari, A.~Barra, P.~Sollich, and L.~Zdeborova.
\newblock Machine learning and statistical physics: theory, inspiration, application.
\newblock {\em J. Phys. A: Math. $\&$ Theor. Special Issue},  2020.

\bibitem{MultiTasking}
E.~Agliari,  et al. 
\newblock Multitasking associative networks. 
\newblock {\em Phys. Rev. Lett.} 109.26:268101, 2012.

\bibitem{Coolen1} 
E.~Agliari,  et al. 
\newblock Immune networks: multi-tasking capabilities at medium load. 
\newblock{\em J. Phys. A: Math. $\&$ Theor.} 46.33:335101, 2013.

\bibitem{Coolen2} 
E.~Agliari,  et al.
\newblock Immune networks: multi-tasking capabilities  near saturation. 
\newblock {\em J. Phys. A: Math. $\&$ Theor.} 46.41:415003, 2013.

\bibitem{EPL-NOI-12}
E.~Agliari,  et al. 
\newblock Retrieving infinite numbers of patterns in a spin-glass model of immune networks
\newblock {\em Europhysics Letters} 117.2: 28003, 2017.


\bibitem{EliGuerra}
E.~Agliari,  et al. 
\newblock Generalized Guerra’s interpolation schemes for dense associative neural networks. 
\newblock {\em Neural Nettworks} 128:254-267,  2020.

\bibitem{FedeParallelo}
E.~Agliari,  et al.  
\newblock Parallel learning by multitasking neural networks.
\newblock {\em JSTAT} 2023.11:113401,  2023.

\bibitem{Albanese2021}
L.~Albanese, F.~Alemanno, A.~Alessandrelli, and A.~Barra. 
\newblock Replica symmetry breaking in dense hebbian neural networks. 
\newblock {\em J. Stat. Phys.} 189(2):1--41, 2022.

\bibitem{NoiInverseFreezing}
L.~Albanese, A.~Barra, E.N.M.~Cirillo. 
\newblock Guerra interpolation for {\em inverse freezing}. 
\newblock {\em arXiv preprint} 2505.06202, 2025.

\bibitem{Amit} D.J. Amit, 
\newblock Modeling brain function: the world of attractor neural networks. 
\newblock {\em Cambridge Press} 1985.

\bibitem{BarraJSP2008}
A.~Barra. 
\newblock The mean-field Ising model trough interpolating techniques. 
\newblock {\em J. Stat. Phys.} 132(5):787-809, 2008.

\bibitem{barra2010replica}
A.~Barra, A.~Di~Biasio, and F.~Guerra.
\newblock Replica symmetry breaking in mean-field spin glasses through the {H}amilton--{J}acobi technique.
\newblock {\em Journal of Statistical Mechanics: Theory and Experiment}, 2010(09):P09006, 2010.

\bibitem{GuerraNN}
A.~Barra, G.~Genovese, and F.~Guerra.
\newblock The replica symmetric approximation of the analogical neural network.
\newblock {\em Journal of Statistical Physics}, 140:784--796, 2010.

\bibitem{B-war1}
A.~Barra, G.~Genovese, and F.~Guerra.
\newblock Equilibrium statistical mechanics of bipartite spin systems.
\newblock {\em Journal of Physics A}, 44(24):245002, 2011.

\bibitem{Dualita1}
A.~Barra, A.~Bernacchia, E.~Santucci, P.~Contucci.  
\newblock On the equivalence of Hopfield networks and Boltzmann machines. 
\newblock {\em Neural Networks}, 34, 1-9, 2012.

\bibitem{B-war4}
A.~Barra, F.~Guerra, and E.~Mingione.
\newblock Interpolating the {S}herrington–{K}irkpatrick replica trick.
\newblock {\em Philosophical Magazine}, 92(1):78, 2012.

\bibitem{blume1971ising}
M. Blume, V. J. Emery, and R. B. Griffiths. 
\newblock {\em Ising model for the $\lambda$ transition and phase separation in He $3-$He $4$ mixtures.}
\newblock { Physical review A}, 4.3 (1971): 1071.

\bibitem{BlumeCambelNN1}
D.~Bolle, I.~Perez-Castillo, and G.~Shim.
\newblock Optimal capacity of the {B}lume-{E}mery-{G}riffiths perceptron.
\newblock {\em Physical Review E}, 67:036113, 2003.

\bibitem{BlumeCambelNN2}
D.~Bolle and T.~Verbeiren.
\newblock Thermodynamics of fully connected {B}lume–{E}mery–{G}riffiths neural networks.
\newblock {\em Journal of Physics A}, 36:295, 2003.

\bibitem{bovier2012mathematical}
A.~Bovier and P.~Picco.
\newblock Mathematical aspects of spin glasses and neural networks.
\newblock {\em Springer Science $\&$ Business Media},  2012.

\bibitem{Bovier1} 
A.~Bovier,  et al. 
\newblock Metastability of Glauber dynamics with inhomogeneous coupling disorder.  
\newblock {\em arXiv preprint}: 2209.09827 (2022).

\bibitem{Bovier2}
A.~Bovier, V.~Gayrard. 
\newblock Metastates in the Hopfield model in the replica symmetric regime. 
\newblock {\em Math. Phys., Analysis and Geometry} 1.2: 107-144, 1998.

\bibitem{CarmonaWu}
P. Carmona and Y. Hu. 
\newblock {\em Universality in {S}herrington–{K}irkpatrick's spin glass model.}
\newblock { Annales de l'Institut Henri Poincare (B) Probability and Statistics}, Vol. 42. No. 2. No longer published by Elsevier, 2006.

\bibitem{Carleo}
G.~Carleo et al., {\em Machine learning and the physical sciences}, Rev. Mod. Phys. \textbf{91}.4:045002, (2019).

\bibitem{Coolen}
A.C.C. Coolen, R. Kuhen, P. Sollich, {\em Theory of neural information processing systems}, Oxford Press (2005).


\bibitem{crisanti2005stable}
A.~Crisanti and L.~Leuzzi, {\em Stable solution of the simplest spin model for inverse freezing}, Phys. Rev.  Lett. \textbf{95}(8):087201, (2005).
 
\bibitem{crisanti-leuzziFullRSB_PRB}
A.~Crisanti and L.~Leuzzi, {\em Thermodynamic properties of a full-replica-symmetry-breaking ising spin glass on lattice gas: The random {B}lume-{E}mery-{G}riffiths-{C}apel model}, Phys. Rev B \textbf{70}:014409, (2014).

\bibitem{Gallius}
G.~Vinci, A.~Galluzzi, M.~Mattia. 
\newblock Beyond catastrophic forgetting in associative networks with self-interactions.  \newblock {\em arXiv preprint} arXiv:2504.04560,  2025.


\bibitem{GalliusTwo}
C.~Wang, et al. 
\newblock Formation of autapse connected to neuron and its biological function
\newblock {\em Complexity} 2017.1:5436737, 2017.

\bibitem{Genovese} 
G. Genovese
\newblock {\em Universality in bipartite mean field spin glasses.}, Journal of mathematical physics 53.12 (2012).

\bibitem{ghatak1977crystal}
S.~Ghatak and D.~Sherrington, {\em Crystal field effects in a general {S} ising spin glass}, J. Phys. C \textbf{10}(16):3149, (1977).



\bibitem{guerra_broken}
F.~Guerra, {\em Broken replica symmetry bounds in the mean field spin glass model}, Comm. Math. Phys. \textbf{233}:1--12, (2003).


\bibitem{B-war3}
F.~Guerra, et al.,  {\em  About a solvable mean field model of a gaussian spin glass},  J. Phys. A: Math. $\&$ Theor. \textbf{47}(15):155002, (2014).

\bibitem{B-war2}
F.~Guerra, et al. {\em Mean field bipartite spin models treated with mechanical techniques}, Europ. Phys. Journ. B \textbf{87}:1, (2014).

\bibitem{KanterPotts}
I.~Kanter, {\em Potts-glass models of neural networks}, Phys. Rev. A \textbf{37}.7:2739,  (1988).

\bibitem{katayama1999ghatak}
K.~Katayama and T.~Horiguchi, {\em {G}hatak-{S}herrington model with spin {S}},  J. Phys. Soc. Japan \textbf{68}(12):3901--3910, (1999).

\bibitem{Kaufman1990}
M.~Kaufman, M.~Kanner, {\em Random-field {B}lume-{C}apel model: Mean-field theory}, Phys. Rev. B \textbf{42}:2378, (1990).

\bibitem{Kanter1} S.~Hodassman, et al., {\em Efficient dendritic learning as an alternative to synaptic plasticity hypothesis}, Sci. Rep. \textbf{12}.1:6571,  (2022).

\bibitem{HopKro1}
D.~Krotov and J.~J. Hopfield, {\em Dense associative memory for pattern recognition}, Adv. Neur. Inf. Proc. Sys. 1180--1188, (2016).

\bibitem{Hopfield}
J.J.~Hopfield, {\em Neural networks and physical systems with emergent collective computational abilities}, Proc. Natl. Acad. Sci. \textbf{79}.8:2554-2558,  (1982).

\bibitem{HopfieldGraded}
J.J.~Hopfield, {\em Neurons with graded response have collective computational properties like those of two-state neurons}, Proc. Natl. Acad. Sci. \textbf{81}.10:3088-3092, (1984).

\bibitem{leuzzi2007spin}
L.~Leuzzi.
\newblock Spin-glass model for inverse freezing.
\newblock {\em Philosophical Magazine}, 87(3-5):543--551, 2007. 

\bibitem{Raimer91} 
R.~Kuehn, S.~Boes, J.L.~van Hemmen. 
\newblock Statistical mechanics for networks of graded-response neurons.
\newblock {\em Phys. Rev. A}, 43.4:2084, 1991.

\bibitem{kanter}
S.~Hodassman, et al. 
\newblock Efficient dendritic learning as an alternative to synaptic plasticity hypothesis.
\newblock {\em Scientific Reports} 12.1:6571, 2022.

\bibitem{Kanter2}
M.~London, M.~Hausser. 
\newblock Dendritic computation. 
\newblock {\em Annu. Rev. Neurosci.} 28.1:503-532, 2005. 


\bibitem{FreezingHopfield}
C.~Morais, F.~Zimmer, and S.~Magalhaes.
\newblock Inverse freezing in the hopfield fermionic ising spin glass with a transverse magnetic field.
\newblock {\em Physics Letters A}, 375:689--697, 2011. 


\bibitem{InverseNew1}
D.~Bolle, I.~Perez-Castillo, G.M.~Shim.
\newblock Optimal capacity of the Blume-Emery-Griffiths perceptron.
\newblock {\em Phys. Rev. E} 67.3:036113, (2003).


\bibitem{InverseNew2}
D.~Bolle, I.~Perez-Castillo.
\newblock Gardner optimal capacity of the diluted Blume–Emery–Griffiths neural network.
\newblock {\em Physica A} 349.3-4:548-562, 2005. 


\bibitem{InverseNew3}
M.~Loewe, F.~Vermet.
\newblock The storage capacity of the Blume–Emery–Griffiths neural network
\newblock {\em J. Phys. A: Math. $\&$ Gen.} 38.16:3483, 2005.



\bibitem{Dualita4}
M.~Mézard. 
\newblock Mean-field message-passing equations in the Hopfield model and its generalizations  
\newblock {\em Physical Review E} 95.2:022117, 2917.

\bibitem{MPV}
M.~Mézard, G.~Parisi, and M.~A. Virasoro.
\newblock {\em Spin glass theory and beyond: An Introduction to the Replica Method and Its Applications}, volume~9.
\newblock World Scientific Publishing Company, 1987.

\bibitem{sherrington1975solvable}
D.~Sherrington and S.~Kirkpatrick.
\newblock Solvable model of a spin-glass.
\newblock {\em Physical review letters}, 35(26):1792, 1975.

\bibitem{Dualita2}
P.~Sollich, et al. 
\newblock Phase diagram of restricted Boltzmann machines and generalized Hopfield networks with arbitrary priors.
\newblock {\em Physical Review E} 97.2: 022310, 2018.

\bibitem{talagrand2000rigorous}
M.~Talagrand.
\newblock Rigorous low-temperature results for the mean field p-spins interaction model.
\newblock {\em Probability theory and related fields}, 117(3):303--360, 2000.

\bibitem{talagrand2003spin}
M.~Talagrand et~al.
\newblock {\em Spin glasses: a challenge for mathematicians: cavity and mean field models}, volume~46.
\newblock Springer Science \& Business Media, 2003.
 
\bibitem{Treveso}
A.~Treves, {\em Graded-response neurons and information encodings in autoassociative memories}, Phys. Rev. A \textbf{42}.4:2418,  (1990).

\bibitem{Dualita3}
J.~Tubiana, R.~Monasson. 
\newblock Emergence of compositional representations in restricted Boltzmann machines.  \newblock {\em Physical review letters} 118.13: 138301, 2017.

\bibitem{parallel_mio}
Agliari, Elena, et al. {\em Parallel learning by multitasking neural networks}, Journal of Statistical Mechanics: Theory and Experiment \textbf{2023.11} (2023): 113401.


\end{thebibliography}

\end{document}